\documentclass[a4paper,UKenglish,cleveref, autoref, thm-restate, numberwithinsect,final]{lipics-v2021}

\nolinenumbers

\usepackage[T1]{fontenc}

\hideLIPIcs
\usepackage[author=anonymous,marginclue,footnote,final]{fixme}

\FXRegisterAuthor{sp}{asp}{SP}
\FXRegisterAuthor{ls}{als}{LS}

\usepackage{xspace}
\usepackage{tikz-cd}
\usepackage{graphicx}
\usepackage{amsfonts}
\usepackage{amsmath}
\usepackage{amssymb}
\usepackage{mathtools}
\usepackage{stmaryrd}
\usepackage{enumitem}
\usepackage{float}
\usepackage{thmtools}
\usepackage{thm-restate}
\usepackage{bussproofs}
\usepackage{relsize}
\usepackage{listings}
\usepackage[strings]{underscore}
\usepackage{hyperref}
\hypersetup{pdftex,colorlinks=true,allcolors=blue}
\usepackage{cleveref}

\newcommand{\gfresh}{\mathsf{N}}
\newcommand{\bfresh}{\mathsf{B}}
\newcommand{\lfresh}{\mathsf{D}}
\newcommand{\Nom}{\mathsf{Nom}}
\newcommand{\Terms}{\mathcal{T}}
\newcommand{\pfun}{\rightharpoonup}

\newcommand{\newletter}[1]{{\midmid}#1}

\newcommand{\midmid}{\hspace{0.2ex}{\rule[-0.1ex]{0.6pt}{1.65ex}}\hspace{0.2ex}}

\newcommand{\FN}{\mathsf{FN}}

\newcommand{\Names}{\mathbb{A}}
\newcommand{\barNames}{\overline{\Names}}

\newcommand{\alphaeq}{\equiv_\alpha}
\newcommand{\powufs}{\pow_{\mathsf{ufs}}}
\newcommand{\pow}{\mathcal{P}}
\newcommand{\powfs}{\pow_{\mathsf{fs}}}

\newcommand{\mycomment}[1]{}
\newcommand{\fresh}{\mathbin{\#}}
\newcommand{\supp}{\mathsf{supp}}
\newcommand{\fix}{\mathop{\mathsf{fix}}}
\newcommand{\Fix}{\mathop{\mathsf{Fix}}}

\newcommand\restr[2]{#1|_{#2}}
\makeatother

\newcommand{\id}{\mathit{id}}

\newcommand{\ExpTime}{\textsc{ExpTime}\xspace}

\theoremstyle{definition}
\newtheorem{defn}[theorem]{Definition}

\numberwithin{equation}{section}
\newcommand{\autorefexpls}[1]{\renewcommand{\exampleautorefname}{Examples}%
\autoref{#1}  \renewcommand{\exampleautorefname}{Example}%
}

\title{Nominal Tree Automata With Name Allocation}
\author{Simon Prucker}{Friedrich-Alexander-Universität Erlangen-Nürnberg,
Germany}{simon.prucker@fau.de}{https://orcid.org/0009-0000-2317-5565}{}%{Funded by the Deutsche Forschungsgemeinschaft (DFG, German Research Foundation) -- project number 434050016}%SpeQT

\author{Lutz Schr{\"o}der}{Friedrich-Alexander-Universität Erlangen-Nürnberg,
Germany}{lutz.schroeder@fau.de}{https://orcid.org/0000-0002-3146-5906}{}%{Funded by the Deutsche Forschungsgemeinschaft (DFG, German Research Foundation) -- project number 434050016}%SpeQT

\authorrunning{S.\ Prucker and L.\ Schr{\"o}der}

\ccsdesc{Theory of computation-Formal languages and automata theory}

\keywords{Data languages, tree automata, nominal automata, inclusion
  checking}
\funding{funded by the Deutsche Forschungsgemeinschaft (DFG, German
  Research Foundation) – Projektnummer 517924115} % CoNaN

  \EventEditors{Rupak Majumdar and Alexandra Silva}
  \EventNoEds{2}
  \EventLongTitle{35th International Conference on Concurrency Theory (CONCUR 2024)}
  \EventShortTitle{CONCUR 2024}
  \EventAcronym{CONCUR}
  \EventYear{2024}
  \EventDate{September 9--13, 2024}
  \EventLocation{Calgary, Canada}
  \EventLogo{}
  \SeriesVolume{311}
  \ArticleNo{23}

\Copyright{Simon Prucker and Lutz Schr{\"o}der}

\begin{document}

\maketitle              % typeset the header of the contribution

\begin{abstract}
  Data trees serve as an abstraction of structured data, such as XML
  documents. A number of specification formalisms for languages of
  data trees have been developed, many of them adhering to the
  paradigm of register automata, which is based on storing data values
  encountered on the tree in registers for subsequent comparison with
  further data values. Already on word languages, the expressiveness
  of such automata models typically increases with the power of
  control (e.g.\ deterministic, non-deterministic, alternating).
  Language inclusion is typically undecidable for non-deterministic or
  alternating models unless the number of registers is radically
  restricted, and even then often remains non-elementary. We present
  an automaton model for data trees that retains a reasonable level of
  expressiveness, in particular allows non-determinism and any number
  of registers, while admitting language inclusion checking in
  elementary complexity, in fact in parametrized exponential time. We
  phrase the description of our automaton model in the language of
  nominal sets, building on the recently introduced paradigm of
  explicit name allocation in nominal automata.
\end{abstract}
\section{Introduction}\label{sec:intro}
\noindent Letters from infinite alphabets generally serve as an
abstraction of data values in formalisms for the specification and
verification of structured data such as data words or data trees
(e.g.~\cite{NevenEA04}). They have variously been used to represent
data values in XML documents~\cite{NevenEA04}, object
identities~\cite{GrigoreEA13}, parameters of method
calls~\cite{HowarEA19}, or nonces in cryptographic
protocols~\cite{KurtzEA07}. There are, by now, quite a number of
automata models for data languages, including register
automata~\cite{KaminskiFrancez94}, data walking
automata~\cite{ManuelEA16}, and data automata~\cite{BojanczykEA11}. A
typical phenomenon in such models is that expressiveness increases
strictly with the power of control (ranging from deterministic to
alternating models). In such models, the key reasoning problem of
inclusion checking tends to be either undecidable or computationally
very hard unless stringent restrictions are imposed on either the
power of control or on key parameters such as the number of
registers. For instance, inclusion checking of nondeterministic
register automata is undecidable unless one restricts either to
unambiguous automata~\cite{MottetQuaas19,Colcombet15} or to automata
with at most two registers~\cite{KaminskiFrancez94} (no complexity
bound being given in the latter case); inclusion checking for
alternating register automata is undecidable unless one restricts to only
one register, and even then fails to be primitive
recursive~\cite{DemriLazic09}; the inclusion problem of data walking
automata is decidable but at least as hard as Petri net
reachability~\cite{ColcolmbetManuel14arXiv}, which by recent results
is
Ackermann-complete~\cite{Leroux21,CzerwinskiOrlikowski21,LerouxS19};
non-emptiness of data automata~\cite{BojanczykEA11} is decidable but,
again, at least as hard as Petri net reachability; and emptiness of
variable automata~\cite{GrumbergEA12} is undecidable unless one
restricts to the (less expressive) deterministic fragment.

Register-based automata models are often essentially equivalent to
automata models in nominal sets~\cite{Pitts13}; for instance, register
automata with nondeterministic reassignment~\cite{KaminskiZeitlin10}
are equivalent to \emph{nondeterministic orbit-finite
  automata}~\cite{BojanczykEA14}. In the nominal setting, one way to
ameliorate the mentioned degree of computational hardness while
retaining a reasonable level of expressiveness is provided by the
paradigm of explicit \emph{name allocation} in nominal automata, which
first appeared in \emph{regular nondeterministic nominal automata
  (RNNA)}~\cite{SchroderEA17} and has subsequently been used in Büchi
RNNA~\cite{UrbatEA21} and in a linear-time nominal
$\mu$-calculus~\cite{HausmannEA21}. In this paradigm, notions of
freshness are based on $\alpha$-equivalence, i.e.\ renaming of bound
names as in $\lambda$-calculus, which blocks renaming of bound names
into names that have free occurrences later in the word; in terms of
the equivalent register-based setting, this amounts to a lossiness
condition saying that at every point, register contents may
nondeterministically be erased (thus freeing the register). At the
same time, name-allocating models impose finite branching. Inclusion
checking in these models is typically elementary, and in fact has low
parametrized complexities with the parameter being the \emph{degree},
which in the register-based setting corresponds to the number of
registers.

Our present contribution is a nominal non-deterministic top-down tree
automaton model that follows the name allocation paradigm. As our main
result, we show that our model of \emph{regular nominal tree automata}
(RNTA) admits inclusion checking in doubly exponential time, more
precisely in parametrized singly exponential time with the degree as
the parameter (recall that the problem is already \ExpTime-complete
for finite-alphabet non-deterministic top-down tree automata). We thus
obtain an efficiently decidable formalism for the specification data
words that admits full non-determinism and unboundedly many registers.

Omitted proofs can be found in the full
version~\cite{PruckerSchroder24arXiv}.

\subparagraph*{Related Work} We have already mentioned work on
register word automata. Kaminski and Tan~\cite{KaminskiTan08}
introduce top-down and bottom-up register tree automata with or
without non-deterministic reassignment; computational hardness of
inclusion and universality is inherited from register word automata
(while membership and emptiness are decidable in elementary complexity
\textsc{NP} and \ExpTime,
respectively~\cite{SendaEA18}). Without non-deterministic
reassignment, the top-down model and the bottom-up model have
incomparable expressiveness.  Our model of regular nominal tree
automata (RNTA) relates, along the usual
correspondence~\cite{BojanczykEA14}, to non-deterministic top-down
register tree automata without non-deterministic reassignment. Van
Heerdt et al.~\cite{HeerdtEA19} treat bottom-up nominal tree automata
with name allocation in a general algebraic setting, describing
minimization and determinization constructions (without considering
inclusion checking; determinization produces orbit-infinite
automata). Since the finite-branching condition for the bottom-up
variant differs from top-down finite branching as imposed in RNTAs, we
expect similar incomparability as in the register-based setting; we
leave the investigation of this point to future work. A register
automata model for data trees with a different navigation pattern has
been introduced by Jurdzi\'nski and
Lazi\'c~\cite{JurdzinskiLazic07,JurdzinskiLazic11} and extended by
Figueira~\cite{Figueira10}; in this model, the automaton moves
downwards or rightwards on the tree instead of passing copies of
itself to child nodes. The emptiness problem of the alternating model
is decidable but not primitive recursive when the number of registers
is restricted to~$1$; expressiveness is incomparable to Kaminski and
Tan's model~\cite{JurdzinskiLazic11}. One formalism for data trees
that does allow inclusion checking in elementary complexity is the
logic $FO^2(+1,\sim)$~\cite{BojanczykEA09}, whose satisfiablity
problem is in \textsc{3NexpTime}. Register tree automata have been
extended to cover an ordering on the data
values~\cite{Tan14,TorunczykZeume22}.

\pagebreak

\section{Preliminaries}\label{sec:prelims}
We assume basic familiarity with category theory
(e.g.~\cite{AdemekEA90}). We briefly recall some background on nominal
sets (see~\cite{Pitts13} for a textbook treatment), as well as on tree
automata (e.g.~\cite{ComonEA08}), register
automata~\cite{KaminskiFrancez94}, and nominal
automata~\cite{BojanczykEA14}.

\subparagraph*{Group Actions and Nominal sets} We fix a countably
infinite set $\Names$ of \emph{names}. Throughout, we let~$G$ denote
the group of finite permutations on~$\Names$, which is generated by
the permutations $(ab)$ that swap two names $a,b\in\Names$; we write
$\id$ for the identity permutation, and $(-)\cdot(-)$ for the group
operation, i.e.\ applicative-order composition of permutations. A
\emph{$G$-set} consists of a set~$X$ and a left action
$(-)\cdot(-)\colon G\times X\to X$ of~$G$ on~$X$ (subject to
$\id\cdot x=x$ and $\pi\cdot (\pi'\cdot x)=(\pi\cdot\pi')\cdot
x$). Given $G$-sets $X,Y$, a map $f\colon X\to Y$ is
\emph{equivariant} if $f(\pi\cdot x)=\pi\cdot f(x)$ for all $x,y$,
and a subset $S \subseteq X$ is \emph{equivariant} if $\pi\cdot s \in S$
for all $\pi \in G$ and $s \in S$, i.e. $S$ is closed under the group action. The
\emph{orbit} of $x\in X$ is $G\cdot x=\{\pi\cdot x\mid\pi\in G\}$. The
orbits form a disjoint decomposition of~$X$; the $G$-set~$X$ is
\emph{orbit-finite} if it has only finitely many orbits.

We write $\fix(x)=\{\pi\in G\mid\pi(x)=x\} $ for $x\in X$, and
$\Fix(A)=\bigcap_{x\in A}\fix(x)$ for $A\subseteq X$. The set~$\Names$
itself is a $G$-set in a canonical manner. A set $S\subseteq\Names$ is
a \emph{support} of an element $x\in X$ if
\begin{equation*}
  \Fix(S)\subseteq\fix(x),
\end{equation*}
that is, if every permutation that fixes all names in~$S$ also
fixes~$x$, which we understand as~$x$ depending only on the names
in~$S$. 

Then, a $G$-set~$X$ is a \emph{nominal set} if every element of~$X$
has a finite support. It follows that every~$x\in X$ has a
\emph{least} finite support $\supp(x)$, also just called the
\emph{support of~$x$}. A name~$a\in\Names$ is \emph{fresh for~$x$} if
$a\notin\supp(x)$, in which case we write $a\fresh x$. We write $\Nom$
for the category of nominal sets and equivariant maps. Examples of
nominal sets include~$\Names$ itself (with $\supp(a)=\{a\}$ for
$a\in\Names$); the product $X\times Y$ of nominal sets $X,Y$ (with
$\supp(x,y)=\supp(x)\cup\supp(y)$); and the finitely supported
powerset $\powfs X$ of a nominal set~$X$, which consists of the
subsets of~$X$ that have finite support under the pointwise action
of~$G$ on the full powerset. A set $A\subseteq X$ is \emph{uniformly
  finitely supported} if $\bigcup_{x\in A}\supp(x)$ is finite; we
write
$\powufs X=\{A\subseteq X\mid A\text{ uniformly finitely supported}\}$
for the \emph{uniformly finitely supported powerset} of~$X$. If~$X$ is
orbit-finite, then the uniformly finitely supported subsets of~$X$ are
precisely the finite subsets.

An important role in the technical development is played by the
\emph{abstraction functor} $[\Names](-)$ on~$\Nom$. For a nominal
set~$X$, $[\Names]X$ is defined as the quotient $\Names\times X/\sim$
of $\Names\times X$ by the equivalence relation~$\sim$ defined by
\begin{equation*}
  (a,x)\sim(b,y)\quad\text{iff}\quad(ca)\cdot x=(cb)\cdot y\quad\text{for $c\fresh(a,x,b,y)$}.
\end{equation*}
Equivalently, $(a,x)\sim(b,y)$ iff either $(a,x)=(b,y)$ or
$y=(ab)\cdot x$ and $b\fresh x$. We write $\langle a\rangle x$ for the
equivalence class of $(a,x)$ under~$\sim$. Thus, $\langle a\rangle x$
may be read as arising from~$x$ by binding the name~$a$. The
equivalence~$\sim$ then captures $\alpha$-equivalent renaming of the
bound name~$a$; in particular, note that by the alternative
description of~$\sim$, renaming~$a$ into~$b$ in $\langle a\rangle x$
is \emph{blocked} when $b\in\supp(x)$.

\subparagraph*{Nominal automata and register automata} The classical
notion of nondeterministic finite automaton can be transferred
canonically to the category of nominal sets, where finiteness
corresponds to orbit-finiteness; this gives rise to the notion of
\emph{nondeterministic orbite-finite automaton
  (NOFA)}~\cite{BojanczykEA14}. For simplicity, we use the
set~$\Names$ of names as the alphabet (more generally, one can work
with any orbit-finite alphabet). Then, a NOFA consists of an
orbit-finite set~$Q$ of states; an equivariant transition relation
$\Delta\subseteq Q\times\Names\times Q$; an equivariant initial
state~$q_0$ (or more generally a set of initial states); and an
equivariant set $F\subseteq Q$ of final states. NOFAs accept finite
words over~$\Names$, with the notions of run and acceptance defined
exactly as in the classical case.

NOFAs are equivalent to a flavour of \emph{nondeterministic register
  automata with nondeterministic reassignment}, roughly described as
follows (see~\cite{BojanczykEA14} for details). A register automaton
has a finite set~$Q$ of control states; a fixed finite number of
registers, which at any point in time can be either empty or contain a
letter (i.e.\ a name from~$\Names$); an initial control state~$q_0$; a
set~$F$ of final control states; and a transition relation~$\delta$
consisting of triples $(q,\phi,q')$ where $q,q'\in Q$ and where~$\phi$
is a boolean combination of equality constraints concerning register
contents before and after the transition and the current input
letter. For simplicity, we require that all registers are initially
empty. A \emph{configuration} of the automaton consists of a control
state and an assignment of letters to some of the registers (the
others are empty). A \emph{run} of the automaton is a sequence of
configurations starting in the initial configuration (consisting
of~$q_0$ and all registers empty) such that every next configuration
is justified by some transition $(q,\phi,q')$ in the sense that the
boolean combination~$\phi$ of equality constraints is satisfied by the
register contents in the configurations before and after the
transition and by the current input letter. A run is accepting if it ends
in a configuration over a final control state. The language of the
automaton is the set of accepted words. What this means is that the
automaton can, in each transition step, perform any combination of the
following actions as specified by the transition constraint~$\phi$:
store the current letter in a register; copy content among registers;
delete content from a register; nondeterministically store any name in
a register (nondeterministic reassigment). These actions are
conditioned on tests for equality or inequality among the registers
and the input letter. In the correspondence between NOFAs and register
automata, a configuration~$c$ of a register automaton becomes a state
in the corresponding NOFA, whose support contains precisely the
letters stored in the registers in~$c$.

Both nondeterminism and nondeterministic reassignment strictly
increase expressive power. For instance, the language~$L_2$ ``some
letter occurs twice'' can be accepted by a nondeterministic register
automaton but not by a deterministic one; and the language ``the
last letter has not been seen before'' can only be accepted using
nondeterministic reassignment. Also, the nondeterministic model is not
closed under complement: The complement of the above-mentioned
language~$L_2$ is the language ``all letters are distinct'', which
cannot be accepted by a nondeterministic register automaton.

\subparagraph*{Tree automata} Throughout, we fix a finite \emph{ranked
  alphabet} (or \emph{signature})~$\Sigma$, consisting of finitely
many \emph{(function) symbols}, each equipped with an assigned finite
\emph{arity}; we write $f/n\in\Sigma$ to indicate that~$f$ is a symbol
of arity~$n$ in~$\Sigma$. We assume that~$\Sigma$ contains at least
one \emph{constant}, i.e.\ a symbol of arity~$0$. The set
$\Terms(\Sigma)$ of (ground) \emph{$\Sigma$-Terms} is defined
inductively by stipulating that whenever $f/n\in\Sigma$ and
$t_1,\dots,t_n\in\Terms(\Sigma)$, then $f(t_1,\dots,t_n)$
in~$\Sigma$. Terms are regarded as a representation of trees, with
each node of the tree labelled with a symbol from~$\Sigma$ whose arity
determines the number of child nodes.

A (classical, finite-alphabet) \emph{nondeterministic top-down tree
  automaton (NFTA)} (e.g.~\cite{ComonEA08}) over~$\Sigma$ is a tuple
$A = (Q, q_0, \Delta)$ (we elide the fixed signature~$\Sigma$) where
$Q$ is a finite set of \emph{states}, $q_0 \in Q$ is the \emph{initial
  state}, and $\Delta$ is a set of \emph{rewrite rules} or
\emph{transitions} of the form
\begin{equation*}
  q(f(x_1,\dots,x_n)) \to f(q_1(x_1),\dots,q_n(x_n))
\end{equation*}
where $f/n \in \Sigma$, $q, q_1,\dots,q_n \in Q$, and the~$x_i$ are
variables; these rules thus manipulate \emph{extended terms}
containing automata states as unary symbols (with at most one
occurrence of such a symbol per tree branch). Much as usual
(e.g.~\cite{BaaderNipkow98,ComonEA08}), such rewrite rules are applied
within a surrounding context, and with the variables~$x_i$ substituted
with ground terms; we continue to write~$\to$ for the arising rewrite
relation on extended terms. The NFTA $A$ \emph{accepts} a
term $t$ if $q_0(t) \to^* t$. The \emph{language} $L(A)$ is
the set of terms (trees) accepted by~$A$. NFTAs have the
same expressiveness as \emph{bottom-up tree automata}, in which the
rewrite rules propagate automata states from the leaves to the
root. Deterministic top-down tree automata, on the other hand, are
strictly less expressive than NFTA. In the present work, our focus is
on nondeterministic top-down models. The \emph{inclusion problem} for
NFTA, i.e.\ to decide whether $L(A)\subseteq L(\mathcal{B})$
for given NFTAs~$A,\mathcal{B}$, is
\ExpTime-complete~\cite{Seidl90}.

\section{A Nominal View on Data Tree Languages}\label{sec:tree-lang}

\noindent We proceed to introduce the relevant notion of nominal tree
language, viewed as representing a data tree language as outlined in
\autoref{sec:prelims}. Trees are represented as a form of algebraic
terms, and carry data values at inner nodes.  Throughout the technical
development, we fix a finite algebraic \emph{signature}~$\Sigma$,
i.e.\ a finite set of operation symbols~$f,g,\dots$, each equipped
with a finite \emph{arity}. We write $f/n\in\Sigma$ to indicate
that~$f$ is an operation symbol of arity~$n$ in~$\Sigma$.

\begin{defn}[name=Terms, label=def:terms]
  We define the set $\Terms_\Names(\Sigma)$ of \emph{nominal
    $\Sigma$-terms}, or just \emph{($\Sigma$-)terms},~$t$ by the
  grammar
  \begin{equation}\label{eq:terms}
    t::= a.f(t_1,\dots,t_n)\mid \nu a.f(t_1,\dots,t_n)\qquad\qquad (t_1,\dots,t_n\in\Terms_\Names(\Sigma), f/n\in\Sigma, a\in\Names)
  \end{equation}
  For uniformity of notation, we occasionally write $\barNames$ for
  the set $\Names\cup\{\nu a\mid a\in\Names\}$; in particular, all
  terms then have the form $\gamma. f(t_1,\dots,t_n)$ with
  $\gamma\in\barNames$.
\end{defn}
\begin{remark}
  Like in the finite-alphabet case as recalled in
  \autoref{sec:prelims}, the base case of the above definition is the
  one where~$f$ is a constant (i.e.\ has arity~$0$). The case of words
  is recovered by taking~$\Sigma$ to consist of unary operations
  and an end-of-word constant. When viewing terms as trees, we see
  the operations of~$\Sigma$ as spanning the tree structure, with
  every node being labelled with an element of~$\barNames$.  We
  understand terms of the form $a.f(t_1,\dots,t_n)$ as being labelled
  with a free name~$a$, and terms of the form $\nu a.f(t_1,\dots,t_n)$
  as allocating a new name~$a$ with scope $f(t_1,\dots,t_n)$. In the
  latter case, the name~$a$ is bound by the $\nu$-operator, in the
  same style as in the $\pi$-calculus or in nominal Kleene
  algebra~\cite{GabbayCiancia11}, with formal definitions to be
  provided presently. (In work on nominal word automata with name
  allocation~\cite{SchroderEA17}, the notation $\newletter{a}$ has
  been used in place of~$\nu a$.) % The standard case of data words,
  % which are sequences of pairs each consisting of one letter from a
  % finite control alphabet and one letter from an infinite data
  % alphabet (e.g~\cite{Colcombet15}) is recovered as a special case
  % by
  % taking~$\Sigma$ to consist of one unary operation for each control
  % letter and an end-of-word constant.
\end{remark}

\begin{remark}[label=rem:binding-sig]
  For brevity of presentation, we have opted for a setup where every
  inner node either binds a name or carries a free name. A
  generalization that allows inner nodes not labelled with any name is
  encoded into the present setup by means of a fixed free dummy name
  that appears in place of the absent name. We thus do use this
  generalization in the examples. In particular, it allows for trees
  not containing any name, while terms according to
  \autoref{def:terms} always contain at least one name.
\end{remark}
The notion of free name informally used above is formally defined in
the expected way:
\begin{defn}
  The set $\FN(t)$ of \emph{free names} of a term~$t$ is defined
  recursively by the clauses
  \begin{align*}
    \FN(a.f(t_1,\dots,t_n))&=\{a\}\cup\FN(t_1)\cup\dots\cup\FN(t_n)\\
    \FN(\nu a.f(t_1,\dots,t_n))&=(\FN(t_1)\cup\dots\cup\FN(t_n))\setminus\{a\}.
  \end{align*}
\end{defn}
Contrastingly, we refer to the name~$a$ in a term
$\nu a.f(t_1,\dots,t_n)$ as a \emph{bound name}. A term~$t$ is
\emph{closed} if $\FN(t)=\emptyset$. The sets $\barNames$ and
$\Terms_\Names(\Sigma)$ become nominal sets under the expected actions
of~$G$, defined on $\barNames$ by $\pi\cdot a=\pi(a)$ and
$\pi\cdot\nu a=\nu\pi(a)$, and on~$\Terms_\Names(\Sigma)$ recursively
by
\begin{equation*}
  \pi \cdot (\delta. f(t_1,\dots,t_n)) = \pi(\delta). f(\pi \cdot t_1,\dots,\pi \cdot t_n).
\end{equation*}
\noindent From this view, we obtain the expected notion of
$\alpha$-equivalence of terms: % For purposes of the treatment using the
% abstraction functor $[\Names](-)$ (\autoref{sec:prelims}), it is
% worth noting two equivalent formulations of $\alpha$-equivalence. By
% the usual inductive definition,
\emph{$\alpha$-equivalence}~$\equiv_\alpha$ is the least congruence on
terms such that
$\nu a.f(t_1,\dots,t_n)\equiv_\alpha \nu b.f(t'_1,\dots,t'_n)$
whenever
$\langle a\rangle (t_1,\dots,t_n)=\langle b\rangle (t'_1,\dots,t'_n)$,
which means that in $\nu a.f(t_1,\dots,t_n)$,~$a$ can be renamed
into~$b$ provided that~$b$ does not occur in $t_1,\dots,t_n$ (by
temporary renaming of inner bound names, one can then also rename~$a$
into~$b$ if~$b$ is not free in $t_1,\dots,t_n$). We write $[t]_\alpha$
for the equivalence class of a term~$t$ under $\alpha$-equivalence. A
term~$t$ is \emph{clean} if all its bound names are mutually distinct
and not free in~$t$, and \emph{non-shadowing} if on every branch
of~$t$, all bound names are mutually distinct and not free in~$t$ (in
particular, no bound name is ever shadowed in~$t$). Clearly, every
term is $\alpha$-equivalent to a clean one (hence, a forteriori, to a
non-shadowing one).  % We give an equivalent
% recursive characterization in \autoref{def:recursiveAlpha} and prove
% the equivalence of both definitions in \autoref{lem:alphaEquiv}.

% begin{rem}
%  For the sake of readability, we refrain from working more generally
%  with \emph{binding signatures}~\cite{FioreEA99} in place of the
%  explicit combination of $\barNames$ and~$\Sigma$.
% \
% end{rem}
\begin{example}
  Under the extension where names are made optional
  (\autoref{rem:binding-sig}), we can represent $\lambda$-terms as
  terms for the signature
  $\{\mathsf{app}/2,\lambda/1,\mathsf{var}/0\}$. For instance, the
  $\lambda$-term $\lambda a.aa$ is represented as the $\Sigma$-term
  $\nu a.\lambda(\mathsf{app}(a.\mathsf{var},a.\mathsf{var}))$. (Of
  course, there are $\Sigma$-terms not corresponding to any
  $\lambda$-term, such as $\nu a.\mathsf{var}$; in our automaton
  model, it will be no problem to exclude such terms by just letting
  them get stuck.) Similarly, we can represent $\pi$-calculus
  expressions as terms for a suitable signature; we will return to
  this in \autoref{expl:pi}. The notion of $\alpha$-equivalence
  defined above is exactly the standard one in these examples.
\end{example}
\noindent As indicated in the introduction, the notion of
$\alpha$-equivalence determines how we interpret allocating a new name
as `reading a name' in a paradigm where we see bound names as
placeholders for arbitrary free names. For this purpose, we
distinguish between \emph{literal tree languages}, which consist of
terms and hence are subsets of $\Terms_\Names(\Sigma)$, and
\emph{alphatic tree languages}, which consist of $\alpha$-equivalence
classes of terms and hence are subsets of
$\Terms_\Names(\Sigma)\,/\equiv_\alpha$. (The latter generalize the
\emph{bar languages} used in work on nominal word
automata~\cite{SchroderEA17} but in the absence of the bar notation
$\newletter a$, that term seems no longer appropriate.) In both cases,
we say that a language is \emph{closed} if it consists of closed
(equivalence classes of) words only. For brevity, we restrict the
main line of the technical treatment to closed languages. The
notion of bound names representing free names is captured by the
function $d\nu$, which removes all occurrences of~$\nu$ from a term,
and is recursively defined by
\begin{equation*}
  d\nu(\nu a.f(t_1,\dots,t_n)) =
  d\nu(a.f(t_1,\dots,t_n))=a.f(d\nu(t_1),\dots,d\nu(t_n)).
\end{equation*}
Thus, $d\nu$ returns terms without name allocation~$\nu a$. We refer
to such terms as \emph{data trees}; they correspond essentially to the
notion of data tree found in the literature (which is often restricted
to binary trees for simplicity,
e.g.~\cite{KaminskiTan08,JurdzinskiLazic07}), with~$\Names$ serving as
the infinite alphabet of data values. We capture notions of freshness
by imposing different disciplines on variable management in
$\alpha$-renaming and subsequently applying~$d\nu$. Specifically, a
\emph{global freshness semantics}, a \emph{branchwise freshness
  semantics}, and a \emph{local freshness semantics} for alphatic tree
languages~$L$ are embodied in the operators~$\gfresh$,~$\bfresh$,
and~$\lfresh$, respectively, defined by
\begin{align*}
  \gfresh(L) & = \{d\nu(t)\mid t\text{ clean}, [t]_\alpha \in 
         L\} & \bfresh(L) & = \{d\nu(t)\mid t\text{ non-shadowing}, [t]_\alpha \in   L\}
  \\
  \lfresh(L) & = \{d\nu(t)\mid [t]_\alpha \in L\}.
\end{align*}
That is,~$N$ and~$B$ insist on clean or non-shadowing
$\alpha$-renaming, respectively, while~$D$ allows unrestricted
$\alpha$-renaming. The languages $\gfresh(L)$, $\bfresh(L)$, and $\lfresh(L)$ are what
we term \emph{data tree languages}, i.e.\ languages consisting only of
data trees. This makes one additional semantics in comparison to the
case of words~\cite{SchroderEA17} where, of course,~$N$ and~$B$
coincide.
\begin{example}\label{expl:semantics}
  Consider $\Sigma=\{f/2,k/0\}$ and the term
  \begin{equation*}
    t=\nu a.f(\nu b.f(a.k,b.k),\nu b.f(b.k,b.k)).
  \end{equation*}
  Then
  \begin{align*}
    \gfresh(\{t\})&=\{a.f(b.f(a.k,b.k),c.f(c.k,c.k))\mid a,b,c\in\Names, a\neq b, a\neq c,b\neq c\}\\
    \bfresh(\{t\})&=\{a.f(b.f(a.k,b.k),c.f(c.k,c.k))\mid a,b,c\in\Names, a\neq b, a\neq c\}\\
    \lfresh(\{t\})&=\{a.f(b.f(a.k,b.k),c.f(c.k,c.k))\mid a,b,c\in\Names, a\neq b\}.
  \end{align*}
  Thus,~$N$ instantiates bound names by free names that are fresh
  w.r.t.\ the entire tree, while~$B$ only enforces freshness w.r.t.\
  the current branch, and allows siblings to instantiate bound names
  by the same free name (i.e.\ allows $b=c$ in $\bfresh(\{t\}$).
  Finally,~$D$, enforces freshness only where $\alpha$-renaming is
  blocked. In the case of~$t$, renaming~$b$ into~$a$ is blocked in the
  left-hand subterm $\nu b.f(a.k,b.k)$ because of the free occurrence
  of~$a$, while there is no such occurrence in the right-hand subterm
  $\nu b.f(b.k,b.k)$ so that~$b$ can be renamed into~$a$ in that
  subterm; this is why $\lfresh(\{t\})$ requires $a\neq b$ but not
  $a\neq c$. Thus,~$N$ and~$B$ are variants of global freshness as
  found, for instance, in session automata~\cite{BolligEA14},
  while~$D$ is indeed a notion of local freshness: We will see later
  (\autoref{lem:suppEdge}) that the presence of the free name~$a$ in
  $\nu b.f(a.k,b.k)$ implies that at the time of reading~$\nu b$ in a
  run of a regular nominal tree automaton, the name~$a$ is still kept
  in memory, so~$D$ essentially enforces freshness w.r.t.\ currently
  stored names in the spirit of register automata models. As
  indicated in the introduction, we trade some of the expressiveness
  of register automata models for computational tractability. In our
  example, this is apparent in the right-hand subterm
  $\nu b.f(b.k,b.k)$, in which freshness of~$b$ w.r.t.~$a$ cannot be
  enforced under~$D$ because there is no free occurrence of~$a$; as a
  slogan,~$D$ enforces freshness w.r.t.\ stored names only if these
  are still expected to be seen again. In work on nominal word
  automata~\cite{SchroderEA17}, it is shown that this phenomenon
  relates to a lossiness property stating that register contents may
  be non-deterministically lost at any time.
\end{example}
\noindent It turns out that both variants of global freshness remain
essentially equivalent to the original alphatic language, in the sense
that they do not affect language inclusion:

\begin{lemma}[label=lem:injectiveN]
  Both~$N$ and~$B$ preserve and reflect language inclusion: For
  alphatic languages $L_1,L_2$, we have $L_1\subseteq L_2$ iff
  $\gfresh(L_1)\subseteq \gfresh(L_2)$ iff $\bfresh(L_1)\subseteq \bfresh(L_2)$.
\end{lemma}
That is, for purposes of checking language inclusion, it does not
matter whether we consider alphatic languages, their global freshness
semantics, or their branchwise freshness semantics.

\section{Regular Nominal Tree Automata}\label{sec:rnta}

We cast our model of \emph{regular nominal tree automata (RNTAs)} as
an extension of \emph{regular nondeterministic nominal automata}
(RNNAs)~\cite{SchroderEA17}, which differ from NOFAs
(\autoref{sec:prelims}) in essentially two ways: Branching is
restricted to be finite (in NOFAs, the set of successors of a state
only needs to be finitely supported, as implied by equivariance of the
transition relation); this is partially made up for by including
\emph{bound transitions} which read bound names. RNTAs natively accept
alphatic tree languages, which as discussed in \autoref{sec:tree-lang}
may be seen as representing languages of data trees in a number of
ways differing w.r.t.\ notions of freshness: Under global or
branchwise freshness semantcs, RNTAs may be seen as a generalization
of session automata~\cite{BolligEA14}, while under local freshness
semantics, they will be seen to correspond to a subclass of
nondeterministic register tree automata~\cite{KaminskiTan08}
characterized by a lossiness condition (\autoref{rem:lossiness}).  As
indicated in the introduction, we thus incur a decrease in
expressiveness in comparison to the register model, which however buys
elementary complexity of inclusion checking.

\begin{defn}[name=Regular nominal tree automata, label=RNTA]
  A \emph{regular nominal tree automaton (RNTA)} over our fixed
  signature $\Sigma$ is a tuple
  \[
    A = (Q, \Delta, q_0)
  \]
  where $Q$ is a orbit-finite nominal set of \emph{states},
  $q_0 \in Q$ is the \emph{initial state} and $\Delta$ is an
  equivariant set of \emph{rewrite rules} or \emph{transitions} of the
  form
  \[
   q(\gamma.f(x_1,\dots,x_n)) \to \gamma.f(q_1(x_1),\dots,q_n(x_n))
  \]
  where $q_1,\dots,q_n \in Q$, $\gamma \in \overline{\mathbb{A}}$,
  $f/n \in \Sigma$, and the~$x_i$ are variables. When no confusion is
  likely, we just write
  $q(\gamma.f(x_1,\dots,x_n)) \to \gamma.f(q_1(x_1),\dots,q_n(x_n))$
  to indicate that
  $(q(\gamma.f(x_1,\dots,x_n)) \to
  \gamma.f(q_1(x_1),\dots,q_n(x_n)))\in\Delta$.  We impose two
  properties on $\Delta$:
  \begin{itemize}
  \item \emph{$\alpha$-invariance}: For
    $q,q_1,\dots,q_n,q'_1,\dots,q'_n\in Q$, if
    $\langle a\rangle (q_1,\dots,q_n) = \langle b\rangle
    (q'_1.\dots,q'_n)$, then
    \begin{align*}
       q(\nu a.f(x_1,\dots,x_n)) & \to \nu a.f(q_1(x_1),\dots,q_n(x_n)) \quad\text{implies}\\
       q(\nu b.f(x_1,\dots,x_n)) & \to \nu b.f(q'_1(x_1),\dots,q'_n(x_n)) .
    \end{align*}
  \item \emph{Finite branching up to $\alpha$-equivalence}: For
    $q \in Q$ and $f/n\in\Sigma$, the sets
    $\{(a, (q_1,\dots,q_n))$ $\mid
     q(a.f(x_1,\dots,x_n))\to a.f(q_1(x_1),\dots,q_n(x_n))\}$ and
    $\{\langle a\rangle(q_1,\dots,q_n)\mid q(\nu a.f(x_1,\dots,x_n))
    \to \nu a.f(q_1(x_1),\dots,q_n(x_n))\}$ are finite.
  \end{itemize}
  Like in the classical case (\autoref{sec:prelims}), the rewrite
  rules in~$\Delta$ may be applied within a surrounding context and
  with variables substituted by (ground) terms.  A state $q \in Q$
  \emph{accepts} a term $t$ if there exists a sequence of applications
  of the rewrite rules in~$\Delta$, called a \emph{run}, that starts
  from~$q(t)$ and ends in~$t$, symbolized as
  $q(t) \xrightarrow{\smash{*}} t$.  We define the \emph{literal tree
    language} $L(q)$ and the \emph{alphatic tree language}
  $L_\alpha(q)$ \emph{accepted} by~$q$ by
  \begin{equation*}
    L(q) = \{t \in \Terms_\Names(\Sigma)\,\mid \text{~$q$ accepts~$t$}\}
    \qquad
    L_\alpha(q) = \{[t]_\alpha\,\mid t\in L(q)\}
  \end{equation*}
  We put $L(A) = L(q_0)$ and $L_\alpha(A) = L_\alpha(q_0)$, i.e.\ the
  RNTA~$A$ \emph{accepts}~$t$ if its initial state
  accepts~$t$. Moreover,~$A$ \emph{accepts} a data tree~$s$
  \emph{under global, branchwise,} or \emph{local freshness semantics}
  if $s\in \gfresh(L_\alpha(A))$, $s\in \bfresh(L_\alpha(A))$, or
  $s\in \lfresh(L_\alpha(A))$, respectively (cf.\ \autoref{sec:tree-lang}).

  The \emph{degree} of~$A$ is the maximal cardinality of supports of
  states in~$A$.
\end{defn}
\noindent We think of the support of an RNTA state as consisting of
finitely many stored names. In examples, we typically write states in
the form
\begin{equation*}
q(a_1,\dots,a_n)
\end{equation*}
where~$q$ indicates the orbit and $a_1,\dots,a_n$ are stored
names. Thus, the degree of an RNTA corresponds morally to the number
of registers. As an important consequence of finite branching, stored
names can come about only by either inheriting them from a predecessor
state or by reading (i.e.\ binding) them:
\pagebreak
\begin{lemma}[label=lem:suppEdge]
  In an RNTA, we have the following properties:
  \begin{enumerate}
  \item If $q(a.f(x_1,\dots,x_n)) \to a.f(q_1(x_1),\dots,q_n(x_n))$,
    then $\supp(q_i)\cup\{a\}\subseteq \supp(q)$ for $i=1,\dots,n$.
  \item If
    $q(\nu a.f(x_1,\dots,x_n)) \to \nu a.f(q_1(x_1),\dots,q_n(x_n))$,
    then $\supp(q_i)\subseteq \supp(q)\cup\{a\}$ for $i=1,\dots,n$.
  \end{enumerate}
  
\end{lemma}

\begin{corollary}[label=cor:suppFN]
  If state~$q$ accepts term~$t$, then $\FN(t) \subseteq \supp(q)$.
\end{corollary}
\begin{remark}\label{rem:closed}
  For brevity, we restrict the further technical treatment to the case
  where \emph{the initial state has empty support}, which by
  \autoref{cor:suppFN} implies that \emph{the accepted language is
    closed}. We also assume this without further mention in the
  examples, with the possible exception of the dummy name needed in
  examples with unlabelled nodes (cf.\ \autoref{rem:binding-sig}), in
  which the dummy name is assumed to be in the support of all states.
\end{remark}

\begin{example}\label{expl:rnta}
  Let $\Sigma=\{f/2,k/0\}$ (so $\Sigma$-terms are just
  $\barNames$-labelled binary trees).
  \begin{enumerate}[wide]
  \item \label{item:universal} The universal data tree language, i.e.\
    the language consisting of all (non-empty, cf.\
    \autoref{rem:binding-sig}) data trees, is accepted under local
    freshness semantics by the RNTA with only one state~$q$ and
    transitions $q(\nu a.f(x,y))\to \nu a.f(q(x),q(y))$,
    $q(\nu a.k)\to\nu a.k$. On the other hand, it is easy to see that
    the universal data tree language cannot be accepted by an RNTA
    under global or branchwise freshness semantics. Under the latter
    semantics, the above RNTA accepts the language of all data trees
    in which all names are distinct or in which all names found on the
    same branch are distinct, respectively.
  \item The data tree language `the letter at the root of the tree
    (which moreover is not a leaf) reappears in all leaves, but not in
    any other node of the tree' is accepted under local freshness
    semantics by the RNTA with states $q_0$, $q_1(a)$ ($a\in\Names$)
    and transitions
    \begin{gather*}
      q_0(\nu a.f(x,y))   \to \nu a. f(q_1(a)(x),q_1(a)(y))  \\
      q_1(a)(\nu b.f(x,y))  \to \nu b. f(q_1(a)(x),q_1(a)(y)) \qquad
      q_1(a)(a.k)   \to a.k
    \end{gather*}
    where we mean this and all further examples to be implicitly closed
    under equivariance and $\alpha$-invariance.  (Regarding notation,
    read $q_1(a)(y)$ as `state $q_1(a)$ processing term~$y$'.)
  \item The data tree language `some letter appears twice' (which, in
    analogy to the word case~\cite{BojanczykEA14,SchroderEA17}, cannot
    be accepted by a deterministic register tree automaton) is
    accepted by the RNTA with states $q_0,q_1(a),q_2$ ($a\in\Names$)
    and transitions
    \begin{align*}
      q_0(\nu a.f(x,y)) & \to f(q_0(x),q_0(y))
      &
        q_0(\nu a.f(x,y)) & \to f(q_1(a)(x),q_1(a)(y))\\
      q_1(a)(a.f(x,y))& \to a.f(q_2(x),q_2(y)) 
    \end{align*}
    and transitions ensuring that~$q_2$ accepts the universal data
    tree language as in item~\ref{item:universal}. It is easy to see
    that the complement of this languages, while acceptable under
    global freshness semantics (item~\ref{item:universal}) cannot be
    accepted by an RNTA under local freshness semantics.
\end{enumerate}
\end{example}
\begin{example}[Structured data]
  We use
  $\Sigma = \{\texttt{!elem}/2,\,\texttt{\#data}/1,\,\texttt{eof}/0
  \}$ to support an XML-like syntax for structured data. We want to
  recognize the language of $\Sigma$-trees where every occurrence of
  $!\texttt{elem}$ is properly closed by $\texttt{eof}$ in the subtree
  below it, at a leaf as far to the left in the tree as possible under
  the policy that later occurrences of $!\texttt{elem}$ are closed
  further to the left. Occurrences of $\texttt{eof}$ and
  $!\texttt{elem}$ are matched by binding a name at $!\texttt{elem}$
  and labelling the corresponding $\texttt{eof}$ with this name
  (moreover, one unlabelled $\texttt{eof}$ closes the entire term), as
  in the term
  \pagebreak
\begin{lstlisting}
      $\nu$a.!elem(
          $\nu$b.#data(
          $\nu$c.!elem(
            $\nu$d.#data(
            $\nu$d.#data(
            $\nu$d.#data(
            c.eof))),
          $\nu$b.#data(
          a.eof),
      eof)))
\end{lstlisting}
  Under local freshness semantics, the data elements~$b$
  in the $\texttt{\#data}$ fields of this term can be any names
  except~$a$, similarly for~$d$ and~$c$.  This language can accepted
  with the RNTA with states $q_0, q_1(a), q_1(c)$ ($a,c\in\Names$) and
  transitions
  \begin{align*}
  q_0(\nu a.\texttt{!elem})(x_1, x_2) &\to \nu a.\texttt{!elem}(q_1(a)x_1,q_0x_2 )\\
  q_0(\nu b.\texttt{\#data})(x_1) &\to \nu b.\texttt{\#data}(q_0x_1)&
  q_0(\texttt{eof}) &\to \texttt{eof}\\
  q_1(a)(\nu c.\texttt{!elem})(x_1, x_2) &\to \nu c.\texttt{!elem}(q_1(c)x_1,q_1(a)x_2 )\\
  q_1(a)(\nu d.\texttt{\#data})(x_1) &\to \nu d.\texttt{\#data}(q_1(a)x_1)&
  q_1(a)(a.\texttt{eof}) &\to a.\texttt{eof}.%\\
  %q_1(c)(\nu a.\texttt{!elem})(x_1, x_2) &\to \nu a.\texttt{!elem}(q_1(a)x_1,q_1(c)x_2 )%\\
  %q_1(c)(\nu d.\texttt{\#data})(x_1) &\to \nu d.\texttt{\#data}(q_1(c)x_1)%&
  %q_1(c)(c.\texttt{eof}) &\to c.\texttt{eof}\\
  \end{align*} % be closed under $\alpha$-invariance. Then the RNTA $A = (\Sigma, Q, \Delta, q_0)$ accepts all terms 
  % where each $\texttt{!elem}$ is properly closed with an corresponding $\texttt{eof}$ carrying the same (bound) name in the second subterm.
  Notice here that although every state stores at most one name, the
  automaton is able to track an unbounded number of $\mathtt{!elem}$
  markers as it effectively creates copies of itself when reading an
  input tree.
  \end{example}

\begin{example}[$\pi$-Calculus expressions]\label{expl:pi}
  % We use
  % $\Sigma =
  % \{\mathit{par}/2,\,\mathit{new}/1,\,\mathit{wr}/1,\,\w\mathit{ch}/1,\,0/0
  % \}$ to represent the \emph{syntax} (only!) of a small fragment of
  % the $\pi$-calculus, with $\mathit{par}$ standing for parallel
  % composition, $\mathit{new}$ for generation of new name, and
  % $a.\mathit{ch}(b.\mathit{wr}(x))$ indicating that output~$b$ is sent
  % on channel~$a$, then continuing with~$x$. The language of all
  % process expressions consisting of finitely many parallel copies of
  We use $\Sigma = \{\mathit{par}/2,\,\mathit{rw}/1,% \,\mathit{!}/1,
  \,\mathit{ch}/1,\,0/0 \}$, $\mathbb{A} = \{\mathit{\bot},a,b,c\dots\}$
  to represent the syntax (only!) of a small fragment of the
  $\pi$-calculus, with $\mathit{par}$ standing for parallel
  composition,
  % $\mathit{new}$ for generation of new name
  and with $\mathit{ch}$ and $\mathit{rw}$ working in combination to
  represent writing or reading on a channel. Here, we model reading
  and allocation of channel names natively as name binding:
  $a.\mathit{ch}$ communicates on an existing channel~$a$,
  and~$\nu a.\mathit{ch}$ on a newly allocated channel~$a$. Similarly,
  $a.\mathit{rw}$ writes~$a$, while $\nu a.\mathit{rw}$ reads a
  value~$a$. E.g., $\nu a.\mathit{ch}(\nu b.\mathit{wr}(x))$ reads~$b$
  from a newly allocated channel~$a$, and continues with~$x$. Let~$L$
  be the language of all $\Sigma$-terms that are parallel composites
  of $k\ge 1$ processes that each read a name~$b$ from a newly
  allocated channel~$a$ and then may, maybe repeatedly, read a new
  name from~$b$ and use that name as the input channel in the next
  round, before terminating ($0)$. This language is accepted by the
  RNTA with states $q_0$, $q_1$, $q_2(a)$ ($a\in\Names$) and
  transitions
  \begin{align*}
    q_0(\mathit{par}(x,y))&\to \mathit{par}(q_0(x),q_0(y))
    & q_0(\nu a.\mathit{ch}(x)) & \to \nu a.\mathit{ch}(q_1(x))\\
    q_1(\nu a.\mathit{rw}(x)) & \to \nu a.\mathit{rw}(q_2(a)(x))
     & q_2(a)(a.\mathit{ch}(x)) & \to a.\mathit{ch}(q_1(x))\\
     q_2(a)(0) & \to 0.
  \end{align*}
  (Notice that the right hand transitions forget the channel name once
  the channel command has been processed; the name is no longer
  needed, as every channel is used only once.)
%   $Q = \mathbb{A} \cup \mathbb{A}^2$, let further
%   \begin{align*}
% %  \begin{split}
%   \Delta = \{q_0(\mathit{par}(x_1, x_2)) &\to \mathit{par}(q_0(x_1),q_0(x_2) )\\
% %  q_0(0) &\to 0\\
% %  q_0(\mathit{!})(x_1) &\to \mathit{!}(q_0x_1)\\
%   q_0(\nu a.\mathit{ch})(x_1) &\to \nu a.\mathit{ch}(q_1(a)(x_1))\,\\
%   \\
%   a(\mathit{par})(x_1, x_2) &\to \mathit{par}(\mathit{a}x_1,\mathit{a}x_2 )\\
%   a(0) &\to 0\\
% %  a(\mathit{!})(x_1) &\to \mathit{!}(\mathit{a}x_1)\\
%   a(\nu b.\mathit{rw})(x_1) &\to \nu b.\mathit{rw}(\mathit{ab}x_1)\\
%   \\
%   \mathit{ab}(\mathit{par})(x_1, x_2) &\to \mathit{par}(\mathit{ab}x_1,\mathit{ab}x_2 )\\
%   \mathit{ab}(0) &\to 0\\
% %  \mathit{ab}(\mathit{!})(x_1) &\to \mathit{!}(\mathit{ab}x_1)\\
%   \mathit{ab}(b.\mathit{rw})(x_1) &\to b.\mathit{rw}(\mathit{ab}x_1)\\
%   \mathit{ab}(\nu b.\mathit{rw})(x_1) &\to \nu b.\mathit{rw}(\mathit{ab}x_1)\}
%   \end{align*} be closed under $\alpha$-invariance.
%   Then the RNTA $A = (\Sigma, Q, \Delta, \mathit{\bot})$ accepts all terms that possibly open a channel~$a$, write or read name $b$ on~$a$, and then stop~($0$).
\end{example}

\begin{remark}
  In the paradigm of universal coalgebra~\cite{Rutten00}, RNTAs may be
  viewed as coalgebras for the functor~$F$ given by
  \begin{equation}\textstyle
    \label{eq:functorRNTA}
    FX =  \powufs(\sum_{f/n\,\in\,\Sigma}(\Names\times X^n+[\Names] X^n))\text{.}
  \end{equation}
\end{remark}

\section{Name Dropping}\label{sec:name-dropping}

The key to the algorithmic tractability of name-allocating automata
models in general~\cite{SchroderEA17,UrbatEA21,HausmannEA21} is to
ensure that the literal language of an automaton is closed under
$\alpha$-equivalence, so that only boundedly many names need to be
considered in inclusion checking. The problem to be overcome here is
that this property does not hold in general, and needs to be enforced
in a modification of the automaton that preserves the alphatic
language. Specifically, the problem comes about by extraneous names
that do not occur in the remainder of a given word to be processed but
do still occur in the relevant successor state, thus blocking the
requisite $\alpha$-renaming. As a simple example, when an automaton
state~$q$ is processing $\nu a.f(t))$ for $f/1\in\Sigma$ and we have a
matching transition $q(\nu a.f(x))\to \nu a. f(q'(x))$, then it may
happen that $b\notin\FN(t)$, so that~$a$ may be $\alpha$-equivalently
renamed into~$b$ in $\nu a.f(t))$, but $a\in\supp(q')$ so that $a$
cannot be $\alpha$-equivalently renamed into~$b$ in $\nu
a. f(q'(x))$. The solution to this is to extend the automaton by
states that come about by dropping some of the names from the support
of previous states~\cite{SchroderEA17}; in the example, a state~$q''$
that has~$b$ removed from its support but otherwise behaves like~$q'$
will allow for the requisite $\alpha$-renaming of the transition into
$\nu a. f(q''(x))$, and will still be able to accept the remaining
term~$t$ since $b\notin\FN(t)$. We proceed to lay out the details of
this construction, which we dub the \emph{name-dropping modification}.

Following work on nominal Büchi automata~\cite{UrbatEA21}, we first
transform the automaton into one whose state set~$Q$ forms a
\emph{strong} nominal set; we do not need the original definition of
strong nominal set~\cite{TzevelekosThesis} but instead use the
equivalent description~\cite{PetrisanThesis} of strong nominal sets as
being those of the form $\sum_{i\in I}\Names^{\#X_i}$ where the~$X_i$
are finite sets and $\Names^{\#X_i}$ denotes the nominal set of total
injective maps $X_i\to\Names$. We generally write elements of sums
like the above as pairs $(i,r)$, in this case consisting of $i\in I$
and
$r\in\Names^{\#X_i}$. % i.e.\ one where for We first note that via an
% automaton whose states form a strong nominal set. As a matter of
% fact, for each RNTA we can show that there exists an RNTA whose
% states form a strong nominal set and that accepts the same language
% (see \autoref{autoStrong}).
Strong nominal sets thus materialize the intuition that the states of
a nominal automaton consist of a control state (the index~$i$ in the
above sum) and a store configuration assigning names to registers in a
duplicate-free manner.
% properties, they do even more allow for building a good intuition
% for how to interpret the states of an automaton: Strong nominal sets
% are indeed isomorphic to sums of the form
% $\sum_{i\in I}\Names^{\#X_i}$, where $\Names^{\#X_i}$ denotes the
% set of total injective maps from $I$-indexed sets $X_i$ to
% $\Names$. We can view this situation as a set of registers that hold
% distinct names from the alphabet. One state would then correspond to
% a single particular configuration of such a set of registers.

\begin{lemma}[label=autoStrong]
  For every RNTA~$A$, there exists an RNTA~$A'$ whose states form a
  strong nominal set such that $A$ and~$A'$ accept the same literal
  language.
\end{lemma}

% As for RNTA based on strong nominal sets, the set of states can be
% perceived as configurations of registers: While strong nominal sets
% are isomorphic to total injective maps - viewed as sets of registers
% that each duplicate-freely saves an element of the alphabet -
\noindent Our name-dropping modification will now come about by
dropping the requirement that each register is necessarily
occupied. This amounts to working with \emph{partial} injective maps
$r\colon X_i\pfun\Names$, with undefinedness (denoted as $r(x)=\bot$)
indicating that a register is currently empty. We first introduce
notation for restricting such partial maps by dropping some of the
names:

\begin{defn}
  Let~$X$ be a finite set. We write $\Names^{\$X}$ for the set of
  partial injective mas $X\pfun\Names$. Let $r\in\Names^{\$X}$, and
  let $N \subseteq \supp(r)=r[X]$.  Then the partial injective map
  $\restr{r}{N}\in\Names^{\$X}$ defined by $\restr{r}{N}(x)=r(x)$ if
  $r(x)\in N$ and $\restr{r}{N}(x)=\bot$ otherwise is the
  \emph{restriction} of $r$ to $N$ (and $r$ is an \emph{extension}
  of~$\restr{r}{N}$).
\end{defn}
%For the least finite support, we then have $\supp(\restr{r}{N}) = N$.

% \begin{remark}
% By definition, every $N \subseteq \Names$ induces a valid restriction $\restr{r}{N}$ of
% the partial injective map $r$.
% Thus, if we quantify over $\restr{r}{N}$, then we inherently quantify over $N$.
% \end{remark}

\begin{defn}[name=Name-dropping modification, label=def:nameDropMod]
  \sloppy Let $A = (Q, \Delta, q_0)$ be an RNTA such that
  $Q = \sum_{i\in I}\Names^{\#X_i}$ is a strong nominal set. For
  $q=(i,r)\in Q$, we write $\restr{q}{N}=(i,\restr{r}{N})$.  Then the
  \emph{name-dropping modification} of $A$ is the RNTA
  $A_\bot = (Q_\bot, \Delta_\bot, q_0)$ where
  \begin{enumerate}
      \item $Q_\bot = \sum_{i\in I}\Names^{\$X_i}$;
      \item for all $q,q'_1,\dots,q'_n\in Q$, $N\subseteq\supp(q)$,
        $N_i\subseteq\supp(q'_i)\cap N$ ($i=1,\dots,n$), and $a\in N$,
        whenever
        $q(a.f(x_1,\dots,x_n)\to a.f(q'_1(x_1),\dots,q'_n(x_n))$
        in~$A$,, then
        $\restr{q}{N}(a.f(x_1,\dots,x_n))\to
        a.f(\restr{q'_1}{N_1}(x_1),\dots,\restr{q'_n}{N_n}(x_n))$
        in~$A_\bot$; and

      \item for all $q,q'_1,\dots,q'_n\in Q$, $N\subseteq\supp(q)$,
        $a\in\Names$, and $N_i\subseteq\supp(q'_i)\cap (N\cup\{a\})$
        ($i=1,\dots,n$), whenever
        $q(\nu a.f(x_1,\dots,x_n))\to \nu
        a.f(q'_1(x_1),\dots,q'_n(x_n))$ in~$A$ and
        $\langle a\rangle \restr{q'_i}{N_i} = \langle b\rangle q_i''$,
        $i=1,\dots,n$, then
        $\restr{q}{N}(\nu b.f(x_1,\dots,x_n))\to \nu
        b.f(q''_1(x_1),\dots,q''_n(x_n))$ in $A_\bot$.
  \end{enumerate}
\end{defn}
\noindent Notice that clauses defining the transition relation
on~$A_\bot$ are only implications: $A_\bot$ inherits transitions
from~$A$ as long as these are consistent with \autoref{cor:suppFN},
and bound transitions in~$A_\bot$ are subsequently closed under
$\alpha$-invariance. \autoref{lem:bijectiveChar} gives
a full description of the arising transitions in the name-dropping
modification.

\begin{lemma}[label=lem:bijectiveChar]
  Let $A = (Q, \Delta, i)$ be an RNTA with $Q$ strong, and let
  $A_\bot = (Q_\bot, \Delta_\bot, i)$ be its name-dropping
  modification.
  \begin{enumerate}
  \item If
    $\restr{q}{N}(a.f(x_1,\dots,x_n))\to
    a.f(q_1(x_i),\dots,q_n(x_n))$ in $A_\bot$ for some
    $q,q_1,\dots,q_n\in Q$ and $N\subseteq\supp(q)$, then
    each~$q_i$ has the form $q_i=\restr{q_i'}{N_i}$ for some
    $q_i'\in Q$, $N_i\subseteq\supp(q)\cap N$ such that
    $q(a.f(x_1,\dots,x_n))\to a.f(q'_1(x_1),\dots,q'_n(x_n))$ in
    $A$.

  \item If
    $\restr{q}{N}(\nu a.f(x_1,\dots,x_n))\to \nu
    a.f(q_1(x_i),\dots,q_n(x_n))$ in $A_\bot$ for some
    $q,q_1,\dots,q_n\in Q$ and $N\subseteq\supp(q)$, then for each
    $q_i$ there is $q'_i$ and
    $N_i\subseteq\supp(q'_i)\cap(N\cup\{b\})$ such that
    $q(\nu a.f(x_1,\dots,x_n))\to \nu
    b.f(q'_1(x_i),\dots,q'_n(x_n))$ in $A$ and
    $\langle b\rangle\restr{q'_i}{N_i}=\langle a\rangle q_i$.
  \end{enumerate}
\end{lemma}

\noindent Essentially, the degree of the name-dropping modification remains the same because the new states
arise by deleting names from the support of previous states; the
number of orbits increases only by a factor $2^d$, where $d$ is the degree, because there are only
$2^d$ ways to delete names from a support of size $d$. Considering this and
as per the intention of the construction, the name dropping
modification closes an RNTA under $\alpha$-equivalence:

\begin{theorem}[label=th:nameDropAlpha]
  For each RNTA $A$, the name-dropping modification~$A_\bot$ of~$A$ is
  an RNTA that accepts the closure of the literal tree language of $A$
  under $\alpha$-equivalence, and hence the same alphatic tree
  language as~$A$. Moreover, $A_\bot$ has the same degree~$d$ as~$A$,
  and the number of orbits of~$A_\bot$ exceeds that of~$A$ by at most
  a factor~$2^d$.
\end{theorem}

% First, we are going to show that the accepted literal tree language of a name-dropping
% modification is closed under $\alpha$-equivalence. In view of this lemma, we intuitively
% proceed as follows: Since $\alpha$-equivalent terms share their free names, in each state
% we can drop all names from the support that do not appear free anymore without losing
% acceptance of terms. Instead, no name in the support of the state does now block
% $\alpha$-renaming of a rewrite rule in order to accept a proper $\alpha$-equivalent term
% anymore. Formally this means:

% With this construction, we have completed our elaborations for incorporating the crucial
% concept of $\alpha$-equivalence into our novel automaton model. Closure under
% $\alpha$-equivalence will play a significant role in the upcoming chapter, where we will
% discuss language inclusion of RNTA as a prominent decision problem for automata.

\begin{remark}[name=Lossiness,label=rem:lossiness] It is apparent from
  the construction of the name-dropping modification that, in the
  usual correspondence between nominal automata models and
  register-based models~\cite{BojanczykEA14,SchroderEA17}, it
  establishes a lossiness property saying that during any transition,
  letters may nondeterministically be lost from the
  registers. Intuitively speaking, the effect of losing a letter from
  a register is on the one hand that one escapes freshness
  requirements against that letter in successor states, but on the
  other hand progress may later be blocked when the lost name is
  required to be seen in the word; the overall consequence of this
  phenomenon is that distinctness of the current letter~$b$ from a
  letter~$a$ seen previously in the word can only be enforced if~$a$
  is expected to be seen again, as already illustrated in
  \autorefexpls{expl:semantics} and~\ref{expl:rnta}.
\end{remark}

\section{Inclusion Checking}\label{sec:inclusion}

\noindent We conclude by showing that language inclusion of RNTAs is
decidable in elementary complexity, in sharp contrast to the typical
situation in register-based models as discussion in
\autoref{sec:intro}. The algorithm is based on reducing the problem to
language inclusion of classical NFTAs over finite signatures
(\autoref{sec:prelims}), using the name-dropping modification to
ensure closure of literal tree languages under $\alpha$-equivalence
(\autoref{sec:name-dropping}): Using closure under
$\alpha$-equivalence, we can choose a finite set of names such that we
can recognize at least one term from each class of $\alpha$-equivalent
terms and then be sure that we also capture all others.  Using this
set of names, we cut out a finite part of the RNTA in which only the
specified names appear. For these restricted automata, which are just
NFTAs, we can decide language inclusion in \ExpTime. A key step in
this programme is thus the following lemma:
\begin{definition}
  Given a finite set~$S\subseteq\Names$, we write
  $\Terms_S(\Sigma)=\{t\in\Terms_\Names(\Sigma)\mid\supp(t)\subseteq
  S\}$.
\end{definition}
(That is, a term is in~$\Terms_S(\Sigma)$ if all its free and bound
names are in~$S$.)
\begin{lemma}[label=lem:restrictNames]
  Let $n_{\mathit{ar}}$ be the maximal arity of symbols in
  $\Sigma$. Suppose that an RNTA~$A$ of degree~$d_A$ accepts a
  term~$t$, and pick $S \subseteq \Names$ such that
  $|S| = d_A \cdot n_{\mathit{ar}} + 1$. Then~$A$ also accepts some
  term $t'\in \Terms_S(\Sigma)$ such that $t' \equiv_\alpha t$.
\end{lemma}
(Recall that initial states have empty support by our running
assumption; otherwise,~$S$ would need to contain the support of the
initial state.)

% We are further going to recall some terminology from complexity theory and introduce a
% new kind of complexity. A decision problem is called \textit{doubly exponential} if there
% exists a deterministic Turing machine that, given an input of size n, can decide the problem
% in $O(2^{2^n})$ time. We further define:

% \begin{definition}[Parametrised singly exponential]
%   We call a decision problem \emph{parametrised singly exponential} in $n$ if the problem can
%   be solved in \(O(2^{p(n,k_0,\cdots,k_m)})\) time on a deterministic Turing machine where
%   the input size \(p\) grows arbitrarily in $k_0, \cdots, k_m$, but at most polynomially in
%   \(n\).
% \end{definition}

\begin{theorem}[label=th:dec]
  Alphatic tree language inclusion $L_\alpha(A)\subseteq L_\alpha(B)$
  of RNTAs~$A,B$ of degrees $d_A$, $d_B$, respectively, over the fixed
  signature~$\Sigma$ is decidable in doubly exponential time, and in
  fact in parametrized singly exponential time with the degree as the
  parameter, i.e.\ exponential in a function that depends
  exponentially on~$d_A+d_B$ and polynomially on the size of $A,B$.
\end{theorem}
\noindent Here, we understand the size of~$A$ and~$B$ in terms of
standard finitary representations of orbit-finite nominal sets, which
essentially just enumerate the support sizes and symmetry groups of
the orbits (e.g.~\cite{UrbatEA21}). In the complexity analysis, we
assume the sigature to be fixed; if the signature is made part of the
input, then parameter includes also the maximal
arity~$n_{\mathit{ar}}$ of symbols in~$\Sigma$ as in
\autoref{lem:restrictNames}.
\begin{proof}
  The proof uses reduction to a
  finite-alphabet~\cite{SchroderEA17,UrbatEA21}. Again, let
  $n_{\mathit{ar}}$ be the maximal arity of symbols in $\Sigma$, and
  pick $S \subseteq \Names$ such that
  $|S| = d_A \cdot n_{\mathit{ar}} + 1$ % and $\supp(i_A) \subseteq S$
  as required in \autoref{lem:restrictNames}. Put
  $\overline{S} = S \cup \{\nu a \mid a \in S\}$, and let~$B_\bot$ be
  the name-dropping modification of~$B$ as per
  \autoref{th:nameDropAlpha}. Let $(Q_A,\Delta_A,q_0^A)=A$ and
  $(Q_B,\Delta_B,q_0^B)=B$. 
  \begin{enumerate}[wide]
  \item\label{item:restrict-S} Show that
    $L_\alpha(A) \subseteq L_\alpha(B)$ iff
    $L(A) \cap \Terms_S(\Sigma) \subseteq L(B_\bot) \cap
    \Terms_S(\Sigma)$.
    \begin{itemize}
    \item `$\Rightarrow$': By \autoref{th:nameDropAlpha},
      $L(B_\bot)$ is the closure of $L(B)$ under
      $\alpha$-equivalence. Thus, $L(A)\subseteq L(B_\bot)$, and
      hence
      $L(A) \cap \Terms_S(\Sigma) \subseteq L(B_\bot) \cap
      \Terms_S(\Sigma)$.
      
    \item `$\Leftarrow$': Let $[t]_\alpha \in L_\alpha(A)$; we have
      to show that $[t]_\alpha \in L_\alpha(B)$. By definition of
      $L_\alpha(A)$, we have $t' \in L(A)$ such that
      $t' \equiv_\alpha t$, so by \autoref{lem:restrictNames} there
      exists $t'' \in L(A)\cap\Terms_S(\Sigma)$ such that
      $t'' \equiv_\alpha t$. Then $t'' \in L(B_\bot)$ by hypothesis,
      and hence $[t]_\alpha\in L_\alpha(B_\bot)$. By
      \autoref{th:nameDropAlpha}, we obtain
      $[t]_\alpha \in L_\alpha(B)$ as required.
    \end{itemize}
    
  \item\label{item:nfta} By \textbf{\ref{item:restrict-S}}, we are
    left to decide whether
    $L(A) \cap \Terms_S(\Sigma) \subseteq L(B_\bot) \cap
    \Terms_S(\Sigma)$.  Observe that $L(A) \cap \Terms_S(\Sigma)$ and
    $L(B_\bot) \cap \Terms_S(\Sigma)$ are effectively just tree
    languages over the finite signature $\overline{S}\times\Sigma$.
    We construct top-down NFTAs~$A_S$ and~$B_S$ over
    $\overline{S}\times\Sigma$ that restrict~$A$ and~$B_\bot$,
    respectively, to $S$ and accept $L(A) \cap \Terms_S(\Sigma)$ and
    $L(B_\bot) \cap \Terms_S(\Sigma)$, respectively: Put
    $A_S = (Q_{A,S}, \Delta_{A,S}, q_0^A)$ where
    $Q_{A,S} = \{q\in Q_A\mid \supp(q)\subseteq S\}$ and
    $\Delta_{A,S} = \{(q (\gamma.f(x_1,\dots,x_n) \to \gamma.
    f(q_1(x_1),\dots,q_n(x_n))) \in \Delta_A\mid q,q_1,\dots,qn\in
    Q_{A,S}, \gamma \in \overline{S}\}$. The construction of
    $B_S = (Q_{B_\bot,S}, \Sigma\times\overline{S}, \Delta_{B_\bot,S},
    i_{B_\bot})$ is analogous. The automata~$A_S$ and $B_S$ are finite
    because~$A$ and~$B_\bot$ are orbit-finite and each orbit of a
    nominal set contains only finitely many elements with a given
    finite support.
    
    % Let us verify that $A_S, B_S$ are indeed top-down NFTA,
    % i.e. $Q_{A,S},Q_{B,S}$ and $\Delta_{A,S},\Delta_{B,S}$ are
    % finite. This is the case since $Q_A, Q_{B,\bot}$ are
    % orbit-finite, thus only finitely many states differ by more than
    % their occurring names, i.e their support. There are only
    % finitely many states with identical support, hence
    % $Q_{A,S},Q_{B,S}$ are finite. This holds also for
    % $\Delta_{A,S},\Delta_{B,S}$, as for finitely many states, names,
    % signature symbols, and their arity also only finitely many
    % unique rewrite rules exist.
  
    We verify that $A_S$ accepts $L(A) \cap \Terms_S(\Sigma)$,
    i.e. $L(A_S) = L(A) \cap \Terms_S(\Sigma)$; the corresponding
    claim for~$B_\bot$ is analogous. Since all states and rewrite
    rules of $A_S$ are inherited from~$A$, it is immediate that
    $L(A_S) \subseteq L(A) \cap \Terms_S(\Sigma)$; we show the revers
    inclusion. So let $t \in L(A) \cap \Terms_S(\Sigma)$; we have to
    show that $A_S$ accepts $t$. We show more generally that every
    state $q$ of $A_S$ accepts all terms $t\in \Terms_S(\Sigma)$
    that~$q$ accepts in~$A$, and proceed via induction on the length
    of an accepting run.
    
    For the base case, let $q$ accept $t = \gamma.c$ in $A$, where
    $\gamma \in \overline{S}$ because $t \in \Terms_S(\Sigma)$. That
    is, we have $\delta = (q(\gamma.c) \to \gamma.c) \in
    \Delta_A$. The $\delta \in \Delta_{A,S}$ by construction, so~$q$
    accepts~$t$ in~$A_S$.
    
    For the inductive step, let~$q$ accept
    $t = \gamma.f(t_1,\dots,t_n)$ in~$A$. Again,
    $\gamma \in \overline{S}$ because $t \in
    \Terms_S(\Sigma)$. Thus, we have
    $\delta = (q(\gamma.f(x_1,\dots,x_n)) \to
    \gamma.f(q_1(x_1),\dots,q_n(x_n))) \in \Delta_A$ such that
    $q_i$ accepts $t_i$ for $i=1,\dots,n$. We distinguish between
    bound and free transitions:
    \begin{itemize}
    \item $\gamma = a$: By \autoref{lem:suppEdge},
      $\supp(q_i) \subseteq \supp(q) \subseteq S$, so $q_i\in Q_S$,
      and by induction,~$q_i$ accepts~$t_i$ in~$A_S$ for
      $i=1,\dots,n$. Since, $\gamma\in S$, we thus have
      $\delta \in \Delta_{A,S}$; it follows that~$q$ accepts~$t$
      in~$A_S$.
    \item $\gamma = \nu a$: By
      \autoref{lem:suppEdge},
      $\supp(q_i) \subseteq \supp(q)\cup \{a\}$. Since $a \in S$ and
      $\supp(q) \subseteq S$, this implies $\supp(q_i) \subseteq S$,
      i.e.\ $q_i\in Q_S$ for $i=1,\dots,n$. By induction,~$q_i$
      accepts~$t_i$ in~$A_S$, and again, $\delta\in \Delta_{A,S}$ by
      construction because $\gamma\in S$, implying that~$q$ accepts~$t$
      in~$A_S$.
    \end{itemize}
      
      % Since the statement holds for every state of $A_S$, it holds in
      % particular for the initial state $q_0^A$ and thus $A_S$ accepts
      % $L(A) \cap \Terms_S(\Sigma)$. For $B_S$ the proof is analogous.
      
  \item\label{item:size} So far, we have reduced the problem to
    deciding language inclusion of NFTAa, which is in \ExpTime
    \cite{ComonEA08}; it remains to analyse the size of the NFTAs
    $A_S$, $B_S$ constructed in step~\textbf{\ref{item:nfta}.}, where
    we have first constructed the name-dropping modification~$B_\bot$
    of the RNTA~$B$ and have then extracted~$A_S$ and~$B_S$ from~$A$
    and~$B_\bot$, respectively, by restricting to the finite set~$S$
    of names. We assume for simplicity that the state spaces of~$A$
    and~$B$ are given as strong nominal sets, so that the size of~$A$
    and~$B$ is essentially the respective number of orbits. When
    estimating the size of~$A_S$ and~$B_S$, it suffices to consider
    the number of states, since the size of the
    signature~$\overline{S}\times\Sigma$ is linear in~$d_A$
    (as~$\Sigma$ is assumed to be fixed) so that the number of
    transitions of NFTAs over~$\overline{S}\times\Sigma$ is polynomial in
    their number of states and~$d_A$. It thus suffices to show that
    the number of states in~$A_S$ and~$B_S$, respectively, is the
    number of orbits of~$A$ or~$B$, respectively, multiplied by a
    factor that is singly exponential in the degree. Now by
    \autoref{th:nameDropAlpha}, the name-dropping modification step
    for~$B$ increases the number of orbits by an exponential factor in
    the degree~$d_B$ but leaves the degree itself unchanged. Moreover,
    we generally have that every orbit of a given nominal set with
    support size~$m$ has at most~$m!$ elements with a given fixed
    support, so the step from $A,B$ to $A_S,B_S$ indeed incurs only an
    exponential factor in the degree, which proves the claim.
    \qedhere
  \end{enumerate}
\end{proof}
From \autoref{lem:injectiveN}, it is immediate that the same
complexity bound as in \autoref{th:dec} holds also for inclusion
checking of RNTAs under global and branchwise freshness semantics,
respectively (i.e.\ for checking whether
$\gfresh(L_\alpha(A))\subseteq \gfresh(L_\alpha(B))$ or
$\bfresh(L_\alpha(A))\subseteq \bfresh(L_\alpha(B))$,
respectively). We conclude by showing that this remains true under
local freshness semantics. The following observation is key:

\begin{defn}
  We define an ordering $\le$ on~$\barNames$ by $a\le\nu a$ for all
  $a\in\Names$. We then define the ordering~$\sqsubseteq$ on
  $\Terms_\Names(\Sigma)$ recursively by $t\sqsubseteq s$ iff $t,s$
  have the form $t=\gamma.f(t_1,\dots,t_n)$ and
  $s=\delta.f(s_1,\dots,s_n)$ where $\gamma\le\delta$ and
  $t_i\sqsubseteq s_i$ for
  $i=1,\dots,n$. % the least preorder such that
  % $a.f(t_1,\dots,t_n)\sqsubseteq\nu a.f(t_1,\dots,t_n)$ for all
  % $a\in\Names$, $f/n\in\Sigma$,
  % $t_1,\dots,t_n\in\Terms_\Names(\Sigma)$ and such that term formation
  % is monotone w.r.t.~$\sqsubseteq$, i.e.\ $t_i\sqsubseteq t'_i$ for
  % $i=1,\dots,n$ implies
  % $\gamma.f(t_1,\dots,t_n)\sqsubseteq\gamma.f(t'_1,\dots,t'_n)$ for
  % all~$\gamma\in\Names$, $f/n\in\Sigma$. Moreover, 
  For a literal language $L$,
  ${\downarrow}L \, = \{t \in \Terms_\Names(\Sigma) \mid \exists t'
  \in L.\, t \sqsubseteq t'\}$ denotes the downward closure of~$L$
  with respect to $\sqsubseteq$.

\end{defn}
That is, $t \sqsubseteq t'$ if $t$ arises from~$t$' by removing zero
or more occurrences of~$\nu$; e.g.\
$\nu a.f(a.f(a.f(a.k))\sqsubseteq \nu a.f(\nu a.(\nu a.f( a.k))$.
\begin{lemma}\label{lem:data-inclusion}
  For closed alphatic languages $L_1,L_2$, we have
  $D(L_1)\subseteq D(L_2)$ iff for all $[t]\in L_1$ there exists
  $t\sqsubseteq t'$ such that $[t']\in L_2$.
\end{lemma}

% We recursively define the relation $\sqsubseteq \,= \{(t = \gamma.f(t_1,\dots,t_n),t' = \delta.f(t'_1,\dots,t'_n))
% \in \Terms_\Names(\Sigma)^2\,|\, (\gamma = a \Rightarrow \delta = \nu a \vee \delta = a) \wedge (\gamma = \nu a \Rightarrow \delta = \nu a) \wedge \forall i \in I. t \sqsubseteq t'\}$,
% 

% \begin{lemma}[label=lem:dataliteral]
%   Let $A$ and $B$ be two RNTA, then
%   $\lfresh(L_\alpha(A)) \subseteq \lfresh(L_\alpha(B))$ iff
%   $L(A) \cap \Terms_S(\Sigma) \subseteq (L(B_\bot) \cap
%   \Terms_S(\Sigma))$.
% \end{lemma}

\begin{theorem}[label=th:litDec]
  Language inclusion $D(L_\alpha(A))\subseteq D(L_\alpha(B))$ under
  local freshness semantics of RNTAs~$A,B$ of degrees $d_A$, $d_B$,
  respectively, over the fixed signature~$\Sigma$ is decidable in
  doubly exponential time, and in fact in parametrized singly
  exponential time with the degree as the parameter, i.e.\ exponential
  in a function that depends exponentially on~$d_A+d_B$ and
  polynomially on the size of $A,B$.
\end{theorem}

\begin{proof}
  The proof is largely analogous to that of \autoref{th:dec}.  In
  step~\textbf{\ref{item:restrict-S}.}, one shows using
  \autoref{lem:data-inclusion} (and recalling \autoref{rem:closed})
  that $D(L_\alpha(A))\subseteq D(L_\alpha(B))$ iff
  $L(A)\cap\Terms_S(\Sigma)\subseteq{\downarrow}(L(B_\bot)\cap\Terms_S(\Sigma))$. In
  step~\textbf{\ref{item:nfta}.}, the NFTA accepting
  ${\downarrow}(L(B_\bot)\cap\Terms_S(\Sigma))$ is constructed as in
  \autoref{th:dec} and then closed downwards under~$\sqsubseteq$ by
  adding a transition
  $q(a.f(x_1,\dots,x_n))\to a.f(q_1(x_1).\dots,q_n(x_n))$ for every
  transition
  $q(\nu a.f(x_1,\dots,x_n))\to \nu a.f(q_1(x_1).\dots,q_n(x_n))$.
\end{proof}

\section{Conclusions}

We have introduced the model of \emph{regular nominal tree automata
  (RNTA)}, a species of non-deterministic top-town nominal tree
automata. RNTAs can be equipped with different data tree semantics
ranging from global freshness as found in session
automata~\cite{BolligEA14} to local freshness. Under the latter, RNTAs
correspond, via the usual equivalence of nominal automata and
register-based automata~\cite{BojanczykEA14,SchroderEA17}, to a
subclass of register tree automata~\cite{KaminskiTan08}. As such, they
are less expressive than the full register model, but in return admit
inclusion checking in elementary complexity (parametrized exponential
time); this in a model that allows unboundedly many registers and
unrestricted nondeterminism (cf.\ \autoref{sec:intro}). RNTAs feature
a native notion of name allocation, allowing them to process terms in
languages with name binding such as the $\lambda$- and the
$\pi$-calculus.

% Upon reflection, our investigation indicates that the transition to
% nominal tree automata with name allocation represents a valid
% generalisation of tree automata and nominal word automata. This
% conceptual extension appears to offer a pragmatic means of
% specifying languages tailored to particular use cases (e.g. XML, the
% $\lambda$- or $\pi$-calculus), without introducing an overly
% escalation in the complexity of relevant decision problems.

Future research will aim in particular at working towards a notion of
nominal automata with name allocation for infinite trees, in
particular with a view to applications in reasoning over
name-allocating fragments of the nominal
$\mu$-calculus~\cite{KlinLelyk19}. Also, it will be of interest to
develop the theory in coalgebraic generality~\cite{Rutten00}, aiming
to support automata with effects such as probabilistic or weighted
branching.

% Acknowledging the potential utility of this approach, it is recognised that there is room for
% further exploration. The general applicability of the method prompts consideration not
% only for generalisation
% of tree automata but also for alternative branching structures featuring name abstraction.
% The prospect of extending the conceptual framework to encompass infinite trees, trees with
% infinite branching, or graph-like branching structures, such as directed acyclic graphs,
% is presented as an avenue for potential exploration. Although further studies and
% exploration is clearly required with an awareness of the inherent complexities, the
% results of our inquiry are presented with a tempered optimism regarding their theoretical
% and practical implications.

\newpage
 \bibliographystyle{plainurl}
 \bibliography{coalgml}
%
%\bibliographystyle{splncs04}
%\bibliography{gsq}
 \newpage \appendix

 \section{Appendix: Additional Details and Omitted Proofs}
 We give details and proofs omitted in the main body.

\subsection{Details for \autoref{sec:tree-lang}}

\begin{lemma}\label{lem:dnu-inj}
  $d\nu$ acts injectively on closed non-shadowing terms.
\end{lemma}
\begin{proof}
  Assume that there are closed non-shadowing terms $t_1,t_2$ such that
  $d\nu(t_1) = d\nu(t_2)$ but $t_1 \neq t_2$. From the definition of
  $d\nu$, we see that~$t_1,t_2$ have the same tree structure and carry
  the same operations and names in every node of the tree, differing
  only in whether the respective name is free or bound. We consider
  positions of tree nodes, and order these positions by stipulating
  that parents are larger than children (so the root becomes the top
  element). Since $t_1 \neq t_2$, there is a \emph{maximal}
  position~$p_1$ at which the respective attached names in~$t_1$
  and~$t_2$ differ w.r.t.\ being free or bound; suppose w.l.o.g.\ that
  this name, say~$a$, is free in~$t_1$ and bound in~$t_2$. Since~$t_1$
  is closed, there must be a position~$p_2$ strictly above~$p_1$ at
  which~$a$ is bound in~$t_1$. By maximality of~$p_1$, $a$ must then
  also be bound at~$p_2$ in~$t_2$; however, this bound occurrence
  of~$a$ is then shadowed at~$p_1$ in~$t_2$, in contradiction to the
  assumption that~$t_2$ is non-shadowing.
\end{proof}

\begin{proof}[Proof of \autoref{lem:injectiveN}]
  We first prove the claim for~$\bfresh$. The `only if' direction (`if
  $L_1\subseteq L_2$, then $\bfresh(L_1)\subseteq \bfresh(L_2)$') is clear. For
  the `if' direction, let~$L_1,L_2$ be alphatic languages such that
  $\bfresh(L_1)\subseteq \bfresh(L_2)$, and let $[t]_\alpha\in L_1$; we have to
  show that $[t]_\alpha\in L_2$. Since every term is
  $\alpha$-equivalent to a non-shadowing one, we can assume that~$t$
  is non-shadowing. Then, $d\nu(t)\in \bfresh(L_1)$, and therefore
  $d\nu(t)\in \bfresh(L_2)$; that is, we have a non-shadowing term~$t'$ such
  that $[t']_\alpha\in L_2$ and $d\nu(t')=d\nu(t)$. But then $t'=t$ by
  \autoref{lem:dnu-inj}, so $[t]_\alpha\in L_2$ as required.

  The claim for~$\gfresh$ is shown completely analogously, using that
  \autoref{lem:dnu-inj} applies in particular also to clean terms and
  that every term is $\alpha$-equivalent to a clean one.
\end{proof}

\subsection{Details for \autoref{sec:rnta}}

\begin{lemma}[name=Equivariance of acceptance, label=lem:equivAcc]
  Let $A = (Q, \Delta, i)$ be an RNTA.
  If $q \in Q$ accepts $t \in \Terms_\Names(\Sigma)$, then for all $\pi \in Perm(\Names)$ $\pi \cdot q$ accepts $\pi \cdot t$.
\end{lemma}

\begin{proof}
    We proceed via induction on the height of terms. Let $A = (Q, \Delta, i)$ be an RNTA, $q \in Q$ a state and $t \in \Terms_\Names(\Sigma)$ a term such that $q$ accepts $t$. 

    For the base case, let $t = (\gamma, c)$ with a constant $c/0 \in \Sigma$ and $\pi \in Perm(\Names)$. Since $q$ accepts $t$, $q(\gamma.c) \to (\gamma, c)\in \Delta$, hence by equivariance $(\pi \cdot q)(\pi(\gamma),c) \to (\pi(\gamma), c)\in \Delta$ and $\pi \cdot q$ accepts $\pi \cdot t = (\pi(\gamma),c)$ by the acceptance condition.

    For the inductive step, let $t = \gamma.f(t_1,\dots,t_n)$ and $\pi \in Perm(\Names)$. Let $\delta = q(\gamma.f(t_1,\dots,t_n)) \to \gamma.f(q_1(t_1),\dots,q_n(t_n)) \in \Delta$ be the first rewrite rule applied of an accepting run. By the induction hypothesis, $\pi \cdot q_i$ accepts $\pi \cdot t_i$ for all $i\in I$, hence it is left to show that $\delta' = (\pi\cdot q)(\pi\gamma).f(t_1,\dots,t_n) \to (\pi\gamma).f((\pi \cdot q_i)(t_1,\dots,t_n) \in \Delta$, which is the case by equivariance of $\Delta$.
\end{proof}

\subsubsection*{Proof of \autoref{lem:suppEdge}}

\begin{proof}
  The proof relies heavily on the finite branching of the automaton.
  \begin{enumerate}
      \item The set $Z = \{(a, (q_1,\dots,q_n))\mid q(a.f(x_1,\dots,x_n)) \to a.f(q_1(x_1),\dots,q_n(x_n))\in \Delta\}$ of free rewrite rules is finite, and for each $z = (a,(q_1,\dots,q_n))\in Z$ $\supp(z) = \{a\}\cup \bigcup_{i \in I} \supp(q_i)$ is finite, thus $S = \bigcup_{z \in Z} \supp(z)$ is a finite support for each $z \in Z$, making $Z$ uniformly finitely supported. For all $z \in Z$ and $i \in I$, it then obviously holds that $\supp(q_i) \cup \{a\} \subseteq \supp(Z)$ and further $\supp(Z) \subseteq \supp(q)$, since $\Delta$ is an equivariant set, $Z$ depends equivariantly on $\Delta$ and equivariant maps do not extend the support.
      \mycomment{
      in particular for each state there are only finitely many free names for which a rewrite rule exists and for each name the number of non-deterministic rewrite rules are finite. Thus necessarily $a \in \supp(q)$, otherwise we would have by equivarince of $\Delta$ infinite free transition between two states: If we have at least one transition, by applying $(ab) \in Perm(\mathbb{})$ for some existing $a$ and some $b$ such that $(ab)$ fixes $\supp(q)$, we would receive one new transition for each name that is not in $\supp(q)$, which are infinite, contradicting finite branching. The same holds for $\supp(q_i) \subseteq \supp(q)$, since otherwise we would have $x \in \supp(q_i)$ such that infinitely many $\pi \in Perm(\Names)$ would fix $\supp(q)$ but not $x$, resulting in infinitely many non-deterministic successor states, again contradicting finite branching.}

      \item \sloppy For bound rewrite rules the argument is analogous: The set $Z = \{\langle a\rangle(q_1,\dots,q_n))\mid q(a.f(x_1,\dots,x_n)) \to a.f(q_1(x_1),\dots,q_n(x_n))\in \Delta\}$ of bound rewrite rules modulo $\alpha$-invariance is finite, and for each $z = \langle a\rangle(q_1,\dots,q_n)\in Z$ $\supp(z) = \bigcup_{i \in I} \supp(q_i) - \{a\}$ is finite, thus $S = \bigcup_{z \in Z} \supp(z)$ is a finite support for each $z \in Z$, making $Z$ uniformly finitely supported. For all $z \in Z$ and $i \in I$, it then also holds that $\supp(q_i) \subseteq \supp(Z) \cup \{a\}$ and further $\supp(Z) \cup \{a\} \subseteq \supp(q) \cup \{a\}$, again since $\Delta$ is an equivariant set, $Z$ depends equivariantly on $\Delta$ and equivariant maps do not extend the support.
      \mycomment{
      For bound transitions the set $\{\langle a\rangle(q_1,\dots,q_n)|q(\nu a.f(x_1,\dots,x_n)) \to \nu a.f(q_1(x_1),\dots,q_n(x_n))\in \Delta\}$ is finite, but we require $\alpha$-invariance so if $q(\nu a.f(x_1,\dots,x_n)) \to \nu a.f(q_1(x_1),\dots,q_n(x_n)) \in \Delta$ and $\langle a\rangle q_i = \langle b\rangle q_i$ then $q(\nu b.f(x_1,\dots,x_n)) \to \nu b.f(q_1(x_1),\dots,q_n(x_n)) \in \Delta$. Thus the argument is the same: by equivariance the only name which $\supp(q_i)$ may contain in addition to ones included in $\supp(q)$ is the one bound in the transition, since it is the only name in which successor states may differ for a state to have infinite transitions.}
  \end{enumerate}
\end{proof}

\subsection{Details for \autoref{sec:name-dropping}}

\subsubsection*{Proof of \autoref{autoStrong}}

\begin{proof}
  Proof as in \cite{UrbatEA21} for Büchi-RNNA with the necessary modifications.
  
  We are left to verify whether the type functor of RNTA, given in \autoref{eq:functorRNTA}, preserves supp-nondecreasing quotients such that the following diagram commutes:
  
  \begin{figure}[H]

      \begin{center}
          \begin{tikzpicture}
              \node (p) at (-3, 1.5) {$P$};
              \node (Fp) at (3, 1.5) {$\mathcal{P}_{\mathit{ufs}}(\mathlarger{\sum}_{f/n \in \Sigma}P^n\times\Names+\mathlarger{\sum}_{f/n\in\Sigma}\lbrack\Names\rbrack P^n)$};
              \draw[->, dotted] (p) -- node [midway, above]{\small$\beta$} (Fp) ;
      
              \node (q) at (-3, -1.5) {$Q$};
              \node (Fq) at (3, -1.5) {$\mathcal{P}_{\mathit{ufs}}(\mathlarger{\sum}_{f/n \in \Sigma}Q^n\times\Names+\mathlarger{\sum}_{f/n\in\Sigma}\lbrack\Names\rbrack Q^n)$};
              \draw[->] (q) -- node [midway, above]{\small$\gamma$} (Fq) ;
      
              \draw[->>] (p) -- node [midway,right]{\small$e$} (q) ;
              \draw[->>] (Fp) -- node [midway,right]{\small$\mathcal{P}_{\mathit{ufs}}(\mathlarger{\sum}_{f/n \in \Sigma}e^n\times\Names+\mathlarger{\sum}_{f/n\in\Sigma}\lbrack\Names\rbrack e^n)$} (Fq) ;
              
          \end{tikzpicture}
      \end{center}
      
  \end{figure}
  
  This is the case since coproducts preserve supp-nondecreasing quotients and the alphabet is a strong nominal set.
  
  \end{proof}
  
  \subsubsection*{Proof of \autoref{lem:bijectiveChar}}

\begin{proof}
  Let $A = (Q, \Delta, i)$ be an RNTA and $A_\bot = (Q_\bot, \Delta_\bot, i)$ its 
  name-dropping modification. States in $A$ have the form $(j, r)$, states in $A_\bot$ instead form $(j, \restr{r}{N})$. In view of \autoref{def:nameDropMod}, there exists a 
  rewrite rule for a specific $r$ extending $\restr{r}{N}$ and we are left to show there 
  exists a rewrite rule for each state of $A$ extending $\restr{r}{N}$:
 \begin{enumerate}
     \item
     Let $r$ extend $\restr{r}{N}$ such that
     $(j,\restr{r}{N})a.f(x_1,\dots,x_n)\to a.f((k_1,\restr{s_1}{N})x_1,\dots,(k_n,\restr{s_n}{N})x_n)
     \in\Delta_\bot$ for some $\restr{s_i}{N}$. We show that 
     $\delta = (j,r)a.f(x_1,\dots,x_n)\to a.f((k_1,s_1)x_1,\dots,(k_n,s_n)x_n)\in \Delta$ for some 
     $s_i$ extending $\restr{s_i}{N}$.
     By construction, there exist some $r',s'_i$ such that 
     $\delta' = (j,r')a.f(x_1,\dots,x_n)\to a.f((k_1,s'_1)x_1,\dots,(k_n,s'_n)x_n)\in \Delta$. By 
     equivariance of $\Delta$, it suffices to find some $\pi \in Perm(\Names)$ such that 
     $\pi \cdot \delta' = \delta$. Choose $\pi$ such that $\pi \cdot r' = r$. Since $Q$ is 
     equivariant and total injective maps have only one orbit, $\pi$ exists, 
     $s = \pi \cdot s'$ also exists and $\pi \cdot \delta' = \delta$. This is the 
     case since by \autoref{lem:suppEdge} $\supp(s)\cup\{a\} \subseteq \supp(r)$, 
     thus $\pi$ fixes $a$.

     \item
     Let $r$ extend $\restr{r}{N}$ such that $(j,\restr{r}{N})\nu b.f(x_1,\dots,x_n)\to
     \nu b.f((k_1,\restr{s_1}{N})x_1,\dots,(k_n,\restr{s_n}{N})x_n)\in\Delta_\bot$ for some $\restr{s_i}{N}$, 
     we show that $\delta = (j,r)\nu a.f(x_1,\dots,x_n)\to \nu a.f((k_1,s_1)x_1,\dots,(k_n,s_n)x_n)
     \in \Delta$ for some $s_i$ such that $\langle b\rangle \restr{(k_i,s_i)}{N} = 
     \langle a\rangle(k_i, \restr{s_i}{N})$ and $a \in \Names$.
     By construction, for some name $c$ there exist some $r', s'_i$ such that $r'$ extends 
     $\restr{r}{N}$, $\langle c\rangle \restr{s_i}{N} = 
     \langle a\rangle\restr{s'_i}{N}$ and $\delta' = 
     (j,r')\nu c.f(x_1,\dots,x_n)\to \nu c.f((k_1,s'_1)x_1,\dots,(k_n,s'_n)x_n)\in \Delta$. Again as for 
     rewrite rules with free names, choose $\pi$ such that $\pi \cdot r' = r$, existing by 
     equivariance. Here, since $\pi$ does not necessarily fix $c$, let $a = \pi(c)$, 
     $s_i = \pi \cdot s'_i$ and eventually $\delta = \pi \cdot \delta' = 
     (\pi \cdot (j,r'))\nu(\pi c).f(x_1,\dots,x_n)\to \nu(\pi c).f((\pi \cdot 
     (k_1,s'_1))(x_1), \dots, (\pi \cdot 
     (k_n,s'_n))(x_n)) = (j,r)\nu a.f(x_1,\dots,x_n)\to 
     \nu a.f((k_1,s_1)x_1,\dots,(k_n,s_n)x_n) \in \Delta$ by equivariance.
     
 \end{enumerate}
\end{proof}

\subsubsection*{Proof of \autoref{th:nameDropAlpha}}
  
We split the proof into a sequence of lemmas:

\begin{corollary}[label=cor:preserveAllRR]
  The name-dropping modification of an RNTA contains all rewrite rules of the unmodified RNTA.
\end{corollary}

\begin{proof}
  Let $A = (Q, \Delta, i)$ be an RNTA and $A_\bot = (Q_\bot, \Delta, i)$ its name-dropping modification.
  
  For free transitions, let $\delta = (j,r)a.f(x_1,\dots,x_n) \to a.f((k_1,s_1)x_1,\dots,(k_n,s_n)x_n) \in \Delta$, we show that $\delta \in \Delta_\bot$. By definition of the name-dropping modification, we instantiate $\restr{r}{N}, \restr{s_i}{N}$ with $r,s_i$, since they extend themselves and by \autoref{lem:suppEdge} $\supp(s) \cup \{a\} \subseteq \supp(r)$.

  For bound transitions, let $\delta = (j,r)\nu a.f(x_1,\dots,x_n) \to \nu a.f((k_1,s_1)x_1,\dots,(k_n,s_n)x_n) \in \Delta$ and we show that $\delta \in \Delta_\bot$. Choose $\restr{r}{N} = r$, $\restr{s_i}{N} = s_i$ and $a = b$, then $\langle b\rangle \restr{s_i}{N} = \langle a\rangle s_i = \langle a\rangle(\restr{s_i}{dom(s_i)}) = \langle a\rangle(\restr{s_i}{dom(\restr{s_i}{N})})$ and $\supp(\restr{s_i}{N}) \subseteq \supp(\restr{r}{N}) \cup \{a\}$ by \autoref{lem:suppEdge}, hence by construction $(j,r)\nu a.f(x_1,\dots,x_n) \to \nu a.f((k_1,s_1)x_1,\dots,(k_n,s_n)x_n) \in \Delta_\bot$.
\end{proof}

\begin{lemma}[label=lem:nameDropModRNTA]
The name-dropping modification of an RNTA is an RNTA.
\end{lemma}

\begin{proof}
Let $A_\bot = (Q_\bot, \Delta_\bot, i)$ be the name-dropping modification of an RNTA $A = (Q, \Delta, i)$. We verify the properties of an RNTA:

\begin{enumerate}
  \item\sloppy The set of rewrite rules $\Delta_\bot$ is equivariant: Let $\delta_\bot = (j, \restr{r}{N})\gamma.f(x_1,\dots,x_n) \to \gamma.f((k_1,\restr{s_1}{N})x_1,\dots,(k_n,\restr{s_n}{N})x_n) \in \Delta_\bot$ and $\pi \in Perm(\Names)$, we have to show that $\delta_\bot' = \pi \cdot \delta_\bot \in \Delta_\bot$. 
  \begin{itemize}
      \item $\gamma = a$: By definition of the name-dropping modification, there is some $\delta = (j, r)a.f(x_1,\dots,x_n) \to a.f((k_1,s_1)x_1,\dots,(k_n,s_n)x_n) \in \Delta$ such that $r,s_i$ extends $\restr{r}{N},\restr{s_i}{N}$, respectively, for all $i \in I$. By equivariance, we have $\delta' = \pi \cdot \delta \in \Delta$. Since the names of $\restr{r}{N}$ appear in $r$ at the same position and likewise for $\restr{s_i}{N}$ and $s_i$ for all $i \in I$, we can restrict the states that appear in $\delta'$ to the states in $\delta_\bot'$, hence by the definition of the name-dropping modification $\delta_\bot' \in \Delta_\bot$.
      \item $\gamma = \nu a$: The situation is analogous to the free case. By definition, it holds that $\delta = (j, r)\nu b.f(x_1,\dots,x_n) \to \nu b.f((k_1,s_1)x_1,\dots,(k_n,s_n)x_n) \in \Delta$ with some $s_i$ such that $\langle b\rangle \restr{s_i}{N} = \langle a\rangle(\restr{s_i}{dom(\restr{s_i}{N})})$ for some $b \in \Names$, $r$ extending $\restr{r}{N}$ and all $i \in I$. By equivariance, we have $\delta' = \pi \cdot \delta \in \Delta$. The names of $\restr{r}{N}$ appear in $r$ at the same position and the same holds for $s_i$ and $\restr{s_i}{N}$, disregarding $a$ and $b$. Since $\langle b\rangle \restr{s_i}{N} = \langle a\rangle(\restr{s_i}{dom(\restr{s_i}{N})})$, $a$ can appear in $s_i$ only at the positions of $\restr{s_i}{N}$ previously undefined for all $i \in I$, which will be discarded later nevertheless, hence we can proceed as for the free case. We can restrict the maps in the states that appear in $\delta'$ to the ones in $\delta'_\bot$ and eventually $\delta'_\bot \in \Delta_\bot$ by $3.$ of \autoref{def:nameDropMod}.     
  \end{itemize}
  
  \item The set of rewrite rules is $\alpha$-invariant: We prove that if $(j, \restr{r}{N})\nu a.f(x_1,\dots,x_n)\to \nu a.f((k_1,\restr{s_1}{N})x_1,\dots,(k_n,\restr{s_n}{N})x_n) \in \Delta_\bot$ and $\langle a\rangle(k_i,\restr{s_i}{N}) = \langle b\rangle(k_i,\restr{s'_i}{N})$ for all $i\in I$, then also $(j, \restr{r}{N})\nu b.f(x_1,\dots,x_n)\to \nu b.f((k_1,\restr{s'_1}{N})(x_1),\dots, (k_n,\restr{s'_n}{N})(x_n)) \in \Delta_\bot$. So let $\delta_\bot = (j, \restr{r}{N})\nu a.f(x_1,\dots,x_n)\to \nu a.f((k_1,\restr{s_1}{N})x_1,\dots,(k_n,\restr{s_n}{N})x_n) \in \Delta_\bot$ and $q'_i \in Q_\bot$ such that $\langle a\rangle(k_i,\restr{s_i}{N}) = \langle b\rangle(k_i,\restr{s'_i}{N})$ for all $i \in I$.
  Choose $r$ extending $\restr{r}{N}$.% such that $b \not\in \supp(r)$.
  Then $\delta = (j, r)\nu c.f(x_1,\dots,x_n)\to \nu c.f((k_1,s_1)x_1,\dots,(k_n,s_n)x_n) \in \Delta$ for some $s_i$ such that $\langle a\rangle \restr{s_i}{N} = \langle c\rangle(\restr{s_i}{dom(\restr{s_i}{N})})$ for $i \in I$ by \autoref{lem:bijectiveChar}. Since also $\langle a\rangle(k_i,\restr{s_i}{N}) = \langle b\rangle(k_i,\restr{s'_i}{N})$ for all $i\in I$, $a$ and $b$ appear in $s_i$ only in the previously undefined part of $\restr{s_i}{N}$ for all $i \in I$, thus further $ \langle c\rangle(\restr{s_i}{dom(\restr{s_i}{N})}) =  \langle a\rangle((ac)\cdot\restr{s_i}{dom(\restr{s_i}{N})}) = \langle a\rangle \restr{s_i}{N} = \langle b\rangle \restr{s'_i}{N}$ for all $i \in I$, hence finally $(j, \restr{r}{N})\nu b.f(x_1,\dots,x_n)\to \nu b.f((k_1,\restr{s'_1}{N})(x_1), \dots, (k_n,\restr{s'_n}{N})(x_n)) \in \Delta_\bot$ by \autoref{def:nameDropMod}.

  \item The automaton is finitely branching: We have to verify that for each state $(j,\restr{r}{N}) \in Q_\bot$ the number of rewrite rules where $(j,\restr{r}{N})$ appears on the left side ("outgoing transitions") is finite, with free names and bound names modulo $\alpha$-invariance respectively, i.e. for all $(j,\restr{r}{N}) \in Q_\bot$ the sets $\{(a, ((k_1,\restr{s_1}{N}),\dots, (k_n,\restr{s_n}{N})))\mid (j,\restr{r}{N})a.f(x_1,\dots,x_n) \to a.f((k_1,\restr{s_1}{N})x_1,\dots,(k_n,\restr{s_n}{N})x_n)\in \Delta_\bot\}$ and $\{\langle a\rangle((k_1,\restr{s_1}{N}),\dots,(k_n,\restr{s_n}{N}))\mid (j,\restr{r}{N})\nu a.f(x_1,\dots,x_n) \to \nu a.f((k_1,\restr{s_1}{N})x_1,\dots,(k_n,\restr{s_n}{N})x_n)\in \Delta_\bot\}$ are finite (we w.l.o.g. assume $A$ to be based on a strong nominal set).
  \begin{itemize}
      \item The set of free rewrite rules $\{(a, ((k_1,\restr{s_1}{N}),\dots, (k_n,\restr{s_n}{N})))\mid (j,\restr{r}{N})a.f(x_1,\dots,x_n) \to a.f((k_1,\restr{s_1}{N})(x_1),\dots, (k_n,\restr{s_n}{N})(x_n))\in \Delta_\bot\}$ is finite: Since $A$ is an RNTA, $\{(a, ((k_1,s_1), \dots, (k_n,s_n)))\mid (j,r)a.f(x_1,\dots,x_n) \to a.f((k_1,s_1)x_1,\dots,(k_n,s_n)x_n)\in \Delta\}$ is finite. We verify that $2.$ from \autoref{def:nameDropMod}, which generates $\Delta_\bot$ from $\Delta$, preserves finiteness of the number of free outgoing transitions also in $\Delta_\bot$. Since for an individual $a \in \Names$ it holds that $\supp((k_i,s_i))$ is finite for each $(k_i,s_i)$ appearing in some $(j,r)a.f(x_1,\dots,x_n) \to a.f((k_1,s_1)x_1,\dots,(k_n,s_n)x_n)\in \Delta$, the set $\{(k_i,\restr{s_i}{N})\}$ of possible states that are extended by $(k_i,s_i)$ is also finite because we only drop names from the support. Furthermore, since by $2.$ of \autoref{def:nameDropMod} $a \in \supp(j,r)$ and $\supp(j,r)$ is finite, also $\{(a, ((k_1,\restr{s_1}{N}),\dots,(k_1,\restr{s_1}{N})))\mid (j,\restr{r}{N})a.f(x_1,\dots,x_n) \to a.f((k_1,\restr{s_1}{N})(x_1), \dots, (k_n,\restr{s_n}{N})(x_n))\in \Delta_\bot\}$ is finite.
      \item The set of bound rewrite rules modulo $\alpha$-invariance $\{\langle a\rangle((k_1,\restr{s_1}{N}),\dots, (k_n,\restr{s_n}{N}))\mid (j,\restr{r}{N})\nu a.f(x_1,\dots,x_n) \to \nu a.f((k_1,\restr{s_1}{N})x_1,\dots,(k_n,\restr{s_n}{N})x_n)\in \Delta_\bot\}$ is finite: Again, $\{\langle a\rangle((k_1,s_1),\dots, (k_n,s_n))\mid (j,r)\nu a.f(x_1,\dots,x_n) \to \nu a.f((k_1,s_1)x_1,\dots,(k_n,s_n)x_n)\in \Delta\}$ is finite since $A$ is an RNTA. We verify that $3.$ from \autoref{def:nameDropMod} does generate $\Delta_\bot$ such that the number of outgoing bound transitions is finite if we abstract from the bound name. The argument is similar to the free case: For each state $(k_i,s_i)$ that appears in some rule in $\Delta$, the number of states $(k_i,\restr{s_i}{N})$ that are extended by $(k_i,\restr{s_i}{N})$ is finite. Eventually, as \autoref{def:nameDropMod} requires $\langle b\rangle \restr{s_i}{N} = \langle a\rangle(\restr{s}{dom(\restr{s_i}{N})})$, also the set $\{\langle a\rangle((k_1,\restr{s_1}{N}),\dots, (k_n,\restr{s_n}{N}))\mid (j,\restr{r}{N})\nu a.f(x_1,\dots,x_n) \to \nu a.f((k_1,\restr{s_1}{N})x_1,\dots,(k_n,\restr{s_n}{N})x_n)\in \Delta_\bot\}$ is finite since $\restr{s_i}{N}$ is injective and thus abstracting from a name results in one equivalence class for each possible restriction of $s_i$ (which are finitely many as mentioned above).  
  \end{itemize}
   
\end{enumerate}

\end{proof}

\begin{lemma}[label=lem:dropStateAcc]
  Let $(j,r)$ and $(j,\restr{r}{N})$ be states of an RNTA $A$ and its name-dropping
  modification $A_\bot$, respectively, such that $r$ extends $\restr{r}{N}$, i.e.
  $(j,\restr{r}{N})$ is generated by dropping names from $(j,r)$, then $\{ t  \in L(j,r)
  \mid  \FN(t) \subseteq codom(\restr{r}{N}) = \supp(j,\restr{r}{N}) \} \subseteq
  L(j,\restr{r}{N})$.
\end{lemma}

\begin{proof}
  We proceed via induction on the height of $t$. The base case is trivial. Let $(j,r) \in Q$, $(j,\restr{r}{N}) \in Q_\bot$ such that $r$ extends $\restr{r}{N}$, $(j, r)$ accepts $t = \gamma.f(t_1,\dots,t_n)$ and $\FN(t) \subseteq codom(\restr{r}{N}) = \supp(\restr{r}{N})$; we have to show that $(j,\restr{r}{N})$ accepts $t$. Since $(j, r)$ accepts $t$, there is $(j,r)\gamma.f(x_1,\dots,x_n) \to \gamma.f((k_1,s_1)x_1,\dots,(k_n,s_n)x_n) \in \Delta$ such that each $(k_i,s_i)$ accepts $t_i$ for all $i \in I$. We therefore have $\FN(t_i) \subseteq \supp((k_i,s_i))$ by \autoref{cor:suppFN}, thus choose $(k_i,\restr{s_i}{N})$ such that $\FN(t_i) = \supp(\restr{s_i}{N})$, which accepts $t_i$ by the induction hypothesis, for all $i \in I$.
  \begin{itemize}
      \item\sloppy $\gamma = a$: Since $a \in \FN(t)$ and, by \autoref{lem:suppEdge} and choice of the injective maps, $\supp(\restr{r}{N}) \supseteq \FN(t) \supseteq \FN(t_i) = \supp(\restr{s_i}{N})$, it holds that $\supp(\restr{r}{N}) \supseteq \{a\}\, \cup\, \supp(\restr{s_i}{N})$ for all $i \in I$, hence by the construction of the name-dropping modification $(j,\restr{r}{N})a.f(x_1,\dots,x_n) \to a.f((k_1,\restr{s_1}{N})x_1,\dots,(k_n,\restr{s_n}{N})x_n) \in \Delta_\bot$ and eventually $(j,\restr{r}{N})$ accepts $t$.
      \mycomment{
       So choose $(k_i,\restr{s_i}{N})$ with the largest support such that $\supp((k_i,\restr{s_i}{N})) \subseteq \supp((j,\restr{r}{N}))$ and $\supp((k_i,\restr{s_i}{N})) \subseteq \supp((k_i,s_i))$, i.e. $\supp((k_i,\restr{s_i}{N})) = \supp((j,\restr{r}{N})) \cap \supp((k_i,s_i))$. Since both $\FN(t) \subseteq \supp(\restr{r}{N})$ and $\FN(t_i) \subseteq \supp((k_i,s_i))$, holds $\FN(t_i) \subseteq \supp(\restr{s_i}{N})$ and thus $(k_i,\restr{s_i}{N})$ accept $t_i$. Finally $r,s_i,\restr{r}{N}$ and $\restr{s_i}{N}$ fulfill all requirements such that  by construction of the name-dropping modification.}
      \item $\gamma = \nu a$: We verify that $(j,\restr{r}{N})\nu a.f(x_1,\dots,x_n) \to \nu a.f((k_1,\restr{s_1}{N})x_1,\dots,(k_n,\restr{s_n}{N})x_n) \in \Delta_\bot$. Since we have $(j,\restr{r}{N})a.f(x_1,\dots,x_n) \to a.f((k_1,\restr{s_1}{N})x_1,\dots,(k_n,\restr{s_n}{N})x_n) \in \Delta$ for the same bound name $a$, we are left to verify that $\supp(\restr{r}{N}) \cup \{a\} \supseteq \supp(\restr{s_i}{N})$ for all $i \in I$. We have chosen $(j,\restr{r}{N})$ and $(k_i,\restr{s_i}{N})$ such that $\supp(\restr{r}{N})\supseteq \FN(t)$ and $\FN(t_i) = \supp(\restr{s_i}{N})$ for all $i \in I$. Suppose w.l.o.g. $a$ is free in some $t_i$, then $\supp(\restr{r}{N}) \cup \{a\} \supseteq \FN(t) \cup \{a\} = \FN(t_i) = \supp(\restr{s_i}{N})$ for all $i \in I$. The required rewrite rule exists in the name-dropping modification by construction and $(j, \restr{r}{N})$ accepts $t$.
      \mycomment{ 
      The case for bound transitions is almost identical to rewrite rules with free names. In fact by \autoref{lem:suppEdge} for bound transitions we may only suppose $\supp(q_i)\subseteq \supp(q)\cup\{a\}$ if
      $q(\nu a.f(x_1,\dots,x_n)) \to \nu a.f(q_1(x_1),\dots,q_n(x_n)) \in \Delta$, but then $a$ does not appear free anyway.}

      \mycomment{By construction exists a transition $(j,\restr{r}{N})a.f(x_1,\dots,x_n) \to a.f((k_1,s_1)x_1,\dots,(k_n,s_n)x_n) \in \Delta_\bot$, left to show that $(j,s)\alpha.f(x_1,\dots,x_n) \to \alpha.f((k_1,\restr{r_1}{N})(x_1),\dots,(k_n,\restr{r_n}{N})(x_n)) \in \Delta_\bot$ such that $\FN(t_i) \subseteq codom(\restr{s'_i}{N})$. By equivariance of $\Delta$ we can find $\pi \in Perm(\Names)$ such that $\pi \cdot s'_i$
      fixes free names, 
      and so - independent from the actual choice of $\alpha$ since it is the same name also in the bound case - by \autoref{def:nameDropMod} $(j,s)\alpha.f(x_1,\dots,x_n) \to \alpha.f((k_1,\restr{r_1}{N})(x_1),\dots,(k_n,\restr{r_n}{N})(x_n)) \in \Delta_\bot$.}
  \end{itemize}
\end{proof}

%\subsubsection*{Proof of Proposition \ref{lem:alphaClosure}}

\begin{lemma}[label=lem:alphaClosure]
  Given an RNTA $A$, the literal tree language of the name-dropping modification
  $A_\bot$ is closed under $\alpha$-equivalence.
\end{lemma}

\begin{proof}

  We proceed via induction on the height of terms. By induction, it suffices to show acceptance of terms that are generated by $\alpha$-equivalent renaming of the upmost bound name. We reinforce the induction hypothesis, so let $(j,r)$ accept $t = \nu a.f(t_1,\dots,t_n)$, $b\#t_i$ and $t'_i = (ab)\cdot t_i$. We have to show that $(j,r)$ then also accepts $t' = \nu b.f(t'_1,\dots,t'_n)$. Since $(j,r)$ accepts $t$, some $(k_i,s_i)$ accepts $t_i$ for all $i\in I$ such that $(j,r)\nu a.f(x_1,\dots,x_n) \to \nu a.f((k_1,s_1)x_1,\dots,(k_n,s_n)x_n) \in \Delta_\bot$. Let $(k_i,\restr{s_i}{N})$ restrict $(k_i,s_i)$ such that $\supp((k_i,\restr{s_i}{N})) = \supp((k_i,s_i))/\{b\}$ for all $i \in I$. Since $(k_i,s_i)$ accepts $t_i$, $\FN(t_i) \subseteq \supp((k_i,s_i))$ and $b \not\in \FN(t_i)$, hence by \autoref{lem:dropStateAcc} $(k_i,\restr{s_i}{N})$ accepts $t_i$ and thus $[t_i]_\alpha \in L_\alpha((k_i,\restr{s_i}{N}))$ for all $i\in I$. By \autoref{lem:equivAcc}, it holds that $[t'_i]_\alpha \in L_\alpha((ab)(k_i,\restr{s_i}{N}))$ and by the induction hypothesis, $t'_i \in L((ab)(k_i,\restr{s_i}{N}))$. We have $(j,r)\nu a.f(x_1,\dots,x_n) \to \nu a.f((k_1,s_1)x_1,\dots,(k_n,s_n)x_n) \in \Delta_\bot$, hence by the construction of the name-dropping modification we necessarily also have $(j,r)\nu a.f(x_1,\dots,x_n) \to \nu a.f((k_1,\restr{s_1}{N})x_1,\dots,(k_n,\restr{s_n}{N})x_n) \in \Delta_\bot$ and by $\alpha$-invariance of the name-dropping modification (see \autoref{lem:nameDropModRNTA}) and freshness of $b$, also $(j,r)\nu b.f(x_1,\dots,x_n) \to \nu b.f((ab)\cdot(k_1,\restr{s_1}{N})(x_1), \dots, (ab)\cdot(k_n,\restr{s_n}{N})(x_n)) \in \Delta_\bot$ for all $i \in I$. As a result, $t' \in L((j,r))$.

  \end{proof}

%\subsubsection*{Proof of Proposition \ref{lem:sameBar}}

  \begin{lemma}[label=lem:sameBar]
  Let $A_\bot$ be the name-dropping modification of an RNTA $A$, then both automata
  accept the same alphatic tree language, i.e. $L_\alpha(A) = L_\alpha(A_\bot)$.
\end{lemma}

\begin{proof}
  \sloppy Let $A = (Q, i, \Delta)$ be an RNTA and $A_\bot = (Q_\bot, i, \Delta_\bot)$ its name-dropping modification. We show mutual language inclusion:
  
  "$\subseteq$": Since name-dropping preserves all rewrite rules (\autoref{cor:preserveAllRR}),
  $A$ is an subautomaton of $A_\bot$, which therefore accepts all terms it previously did.

  "$\supseteq$": Show for each state $(j,\restr{r}{N}) \in Q_\bot$ accepting $t = \gamma.f(t_1,\dots,t_n)$ that each state $(j, r) \in Q$ with $r$ extending $\restr{r}{N}$ accepts some $t' = \beta.f(t'_i) \equiv_\alpha t$. We proceed via induction on the number of applications of rewrite rules in an accepting run.
  
  For the base case let some $(j,\restr{r}{N}) \in Q_\bot$ accept $t = (\gamma, c)$, we show that each $(j,r) \in Q$ with $r$ extending $\restr{r}{N}$ accepts some $t' = (\beta, c) \equiv_\alpha t$. Since $(j,\restr{r}{N})$ accepts $t$, $(j,\restr{r}{N})(\gamma.c) \to (\gamma.c) \in \Delta_\bot$. We distinguish free and bound cases. Let $\gamma = a$, then by \autoref{lem:bijectiveChar}, it holds that $(j,r)(a,c) \to (a,c) \in \Delta$ for each choice of $r$, hence $(j,r)$ in $A$ accepts $t' = t \equiv_\alpha t$. For the bound case, let $\gamma = \nu a$, then by \autoref{lem:bijectiveChar} $(j,r)\nu b.c \to \nu b.c \in \Delta$ for some $b$ and each choice of $r$, hence $(j,r)$ in $A$ accepts $t' = \nu b.c \equiv_\alpha \nu a.c = t$.
  
  \sloppy For the inductive step, let $(j, \restr{r}{N}) \in Q_\bot$ accept $t$ and $\delta = (j,\restr{r}{N})\gamma.f(x_1,\dots,x_n) \to \gamma.f((k_1,\restr{s_1}{N})x_1,\dots,(k_n,\restr{s_n}{N})x_n) \in \Delta_\bot$ be the first rewrite rule applied in an accepting run. By the acceptance condition and the inductive hypothesis, for all $i \in I$ and all states $(k_i,s_i) \in Q$ extending $\restr{s_i}{N}$, $\forall i \in I. (k_i,s_i)$ accepts some $t''_i\equiv_\alpha t_i$. We distinguish between bound and free transitions:
  \begin{itemize}
      \item Case $\gamma = a$: Let $(j,r) \in Q$ with $r$ extending $\restr{r}{N}$.
      By \autoref{lem:bijectiveChar}, a rewrite rule $(j,r)a.f(x_1,\dots,x_n) \to
      a.f((k_1,s'_1)x_1,\dots,(k_n,s'_n)x_n) \in \Delta$ exists for some $s'_i$ extending
      $\restr{s_i}{N}$ for all $i \in I$. By the induction hypothesis, $(k_i,s'_i)$ accepts
      some $t'_i \equiv_\alpha t_i$ for each $i \in I$, thus $(j,r)$ accepts
      $t' = a.f(t'_1,\dots,t'_n)$, which is $\alpha$-equivalent to $t$.
      
      \mycomment{
      By \autoref{lem:equivAcc} $(k_i,s'_i) = \pi \cdot (k,s_i)$ accept $t'_i = \pi \cdot t''_i$, so $(j,r)$ accepts $t' = a.f(t'_1,\dots,t'_n)$.
      
      Indeed $t' \equiv_\alpha t$, which is essentially the case since all $s_i$ and $s'_i$ extend $\restr{s_i}{N}$: As $(k,\restr{s_i}{N})$ accept $t_i$, $\FN(t_i) \subseteq \supp((k,\restr{s_i}{N}))$. Alpha-equivalent terms carry the same free names, so all free names $\FN(t''_i)$ that are necessarily in $\supp((k,s_i))$ are inherited from $(k,\restr{s_i}{N})$. Since all $s_i,s'_i$ extend $s_{\bot_i}$ and thus all $\pi_i \in Perm(\Names)$ such that $s'_i = \pi_i \cdot s_i$ leave all inherited names untouched, $\pi_i$ fix $\FN(t''_i)$ making $t' = a.f(\pi \cdot t''_1,\dots,\pi \cdot t''_n) \equiv_\alpha t$.}
      
      \item Case $\gamma = \nu a$: It is left to show that all $(j,r) \in Q$ with $r$ extending
      $\restr{r}{N}$ accept some $t' = \beta.f(t'_i) \equiv_\alpha t$. By
      \autoref{lem:bijectiveChar}, $(j,r)\nu b.f(x_1,\dots,x_n) \in Delta$ for some
      $s_i$ and some $b \in \Names$ such that $\langle a \rangle \restr{s_i}{N} =
      \langle b\rangle (\restr{s_i}{dom(\restr{s_i}{N})})$ for all $i \in I$. Let further
      $s'_i$ extend $\restr{s_i}{N}$ such that $\langle a\rangle s'_i = \langle b\rangle s_i$,
      then by the induction hypothesis $s'_i$ accepts some $t'_i \equiv_\alpha t_i$ for all
      $i \in I$. Since $\langle a\rangle s'_i = \langle b\rangle s_i$, $a\not\in \supp(s_i)$
      and $b \not\in \supp(s'_i)$, hence $s_i = (ab)\cdot s'_i$ for all $i \in I$. By
      \autoref{lem:equivAcc}, $s_i$ then accepts $(ab)t'_i$ for all $i \in I$, thus $(j,r)$
      accepts $t' = \nu b.f((ab)\cdot t'_i)$.

      We verify that $t$ and $t'$ are in fact $\alpha$-equivalent. Since
      $a\not\in \supp(s_i)$, $b \not\in \supp(s'_i)$ and $s'_i$ accepts $t'_i$, it follows
      that $b \not\in \FN(t'_i)$, thus $a \not\in \FN((ab)\cdot t'_i)$ and further
      $\langle a \rangle t'_i = \langle b \rangle (ab)t'_i$ for all $i \in I$. Since
      $t'_i \equiv_\alpha t_i$, also $\langle a\rangle [t_i]_\alpha =
      \langle b\rangle[(ab)t'_i]_\alpha$ for all $i\in I$, hence eventually
      $t = \nu a.f(t_1,\dots,t_n) \equiv_\alpha \nu b.f((ab)\cdot t'_i) = t'$.

  \end{itemize}

  As we have seen, for each state $(j,\restr{r}{N})$ of the name-dropping modification and each term $t$ that $(j,\restr{r}{N})$ accepts, each state $(j,r)$ from the original RNTA with $r$ extending $\restr{r}{N}$ accepts an $\alpha$-equivalent term $t'$. This holds in particular for the initial state, which is for both the RNTA and the name-dropping modification identical and is based on a total injective map, which extends itself since the extension of maps is a reflexive property, hence the RNTA accepts a term of the $\alpha$-equivalence class of each term accepted by its name-dropping modification, i.e. $L_\alpha(A) \supseteq L_\alpha(A_\bot)$.
\end{proof}

\noindent The proof of \autoref{th:nameDropAlpha} itself is now
immediate: By \autoref{lem:alphaClosure}, the literal language of the
name-dropping modification~$A_\bot$ of an RNTA~$A$ is closed under
$\alpha$-equivalence; on the other hand,~$A_\bot$ accepts, by
\autoref{lem:sameBar}, the same alphatic tree language as~$A$, which
is the set of $\alpha$-equivalence classes of terms literally accepted
by~$A$, so every term literally accepted by~$A_\bot$ is
$\alpha$-equivalent to one accepted by~$A$. \qed

\subsection{Details for \autoref{sec:inclusion}}

\subsubsection*{Proof of \autoref{lem:restrictNames}}

\begin{proof}

  We call two runs of equal structure equivalent if for each pair of two rewrite rules
  \[\delta = q\gamma.f(x_1,\dots,x_n)\to \gamma.f(q_1x_1,\dots,q_nx_n)  \]
  and 
  \[\delta' = q'\epsilon.f(x_1,\dots,x_n)\to \epsilon.f(q'_1x_1,\dots,q'_nx_n)\text{,}\]
  nested at the same position within the structure of their respective run, there exist $\pi, \tau$ such that $\pi q = q'$ and $\tau q_i = q'_i$ for each $i \in I$ (i.e. for each height of the run, there exists one single permutation that transforms the states).
  
  Choose $S \supseteq \supp(q)$ such that $|S| = m \cdot n_{\mathit{ar}} + 1$. We show that for each $q \in Q$ and $t$ such that $q$ accepts $t$, there exists $t'$ such that $q$ accepts $t'$ with an equivalent run and $\mathsf{N}(t') \subseteq S$. We proceed via induction over the height of trees.
  
  For the base case, let $q \in Q$ accept $t = (\gamma, c)$. If $\gamma = a$, then $a \in \supp(q)\subseteq S$ and $t' = t$. If $\gamma = \nu a$, then the required transition exists for all names, in particular for the ones in $S$, since by $\alpha$-invariance for constants all bound transitions exist nevertheless.
  
  For the inductive step, let $q$ accept $t = \gamma.f(t_1,\dots,t_n)$ and $\delta = q(\gamma.f(t_1,\dots,t_n)) \to \gamma.f(q_1t_1,\dots,q_nt_n)$ be the first rewrite rule applied in an accepting run. Then for $t_i$ there exists by induction hypothesis $t'_i \equiv_\alpha t_i$ accepted by $q_i$ with an equivalent run and $\mathsf{N}(t'_i) \subseteq S$ for each $i \in I$. We distinguish between the bound and free case:
  \begin{itemize}
      \item $\gamma = a$: We have $q(a.f(x_1,\dots,x_n)) \to a.f(q_1x_1,\dots,q_nx_n) \in \Delta$ so by \autoref{lem:suppEdge} $\supp(q_i)\cup \{a\} \subseteq \supp(q)$ for all $i \in I$. Let $t' = a.f(t'_1,\dots,t'_n)$. Since $a \in \supp(q) \subseteq S$ and $\mathsf{N}(t'_i) \subseteq S$ for all $i \in I$, $\mathsf{N}(t') \subseteq S$, $t' \equiv_\alpha t$ and $id\cdot q_i = q_i$ and from $q_i$ on we have equivalent runs for each $i \in I$, hence $q$ accepts $t'$ with an equivalent run.
  
      \item $\gamma = \nu a$: We have $q(\nu a.f(x_1,\dots,x_n)) \to \nu a.f(q_1x_1,\dots,q_nx_n) \in \Delta$ so by \autoref{lem:suppEdge} $\supp(q_i) \subseteq \supp(q) \cup \{a\}$ for all $i \in I$. If $a \in S$, we can just proceed as in the free case with $t' = \nu a.f(t_1,\dots,t_n)$, which is $\alpha$-equivalent to $t$ and accepted by $q$ with an equivalent run since the first transition is the same.
      
      So suppose $a\not\in S$. Then $a\#q_i$ for each $i\in I$. Furthermore, since
      $|I| \leq n_{\mathit{ar}}$, for all $i \in I$ $|\supp(q_i)| \leq m$ and $\supp(q_i) \subseteq S$,
      there exists $b \in S$ such that $b\#q_i$ for each $i \in I$. Since both $a\#q_i$ and
      $b\#q_i$, $\langle a\rangle q_i = \langle b\rangle q_i$. So by $\alpha$-invariance
      $q(\nu b.f(x_1,\dots,x_n)) \to \nu b.f(q_1x_1,\dots,q_nx_n) \in \Delta$, $q_i$ accepts
      $t_i$ for each $i \in I$ and thus $q$ accepts $t' = \nu b.f(t_1,\dots,t_n)$. Lastly,
      $t' \equiv_\alpha t$ since $a$ and $b$ are both
      fresh for $t_i$ for all $i \in I$. \qedhere
     \end{itemize}
  \end{proof}

  \subsubsection*{Proof of \autoref{lem:data-inclusion}}
  
  We first note an analogue of~\cite[Lemma~6.4]{UrbatEA21}:
  \begin{lemma}\label{lem:alpha-increase}
    If $t\alphaeq s$ and $t\sqsubseteq t'$, then there exists $s'$
    such that $t'\alphaeq s'$ and $s\sqsubseteq s'$.
  \end{lemma}
  In \autoref{lem:data-inclusion}, the `if' direction is clear; we
  prove `only if'. The latter reduces immediately to the following:
  \begin{lemma}
    Let~$L$ be a closed alphatic language, and let~$t$ be a term
    such that $D(\{[t]_\alpha\})\subseteq D(L)$. Then there
    exists $[t']_\alpha\in L$ such that $t\sqsubseteq t'$.
  \end{lemma}
  \begin{proof}
    Since~$L$ is closed,~$t$ is closed. By
    \autoref{lem:alpha-increase}, we can replace~$t$ with any
    $\alpha$-equivalent term, so we can assume w.l.o.g.\ that~$t$ is
    clean. We use the terminology of positions in trees as in the
    proof of \autoref{lem:dnu-inj}. By assumption, we have
    $d\nu(t)\in D(L)$, so there exists $[t']\in L$ such that
    $d\nu(t')=d\nu(t)$. We claim that $t\sqsubseteq t'$. Since
    $d\nu(t')=d\nu(t)$,~$t$ and $t'$ have the same tree structure and
    have inner nodes carrying the same signature symbols, and differ
    only in whether the attached names are free or bound. Suppose
    that~$t$ carries a bound name $\nu a$ at position~$p_1$; we have
    to show that $t'$ also carries~$\nu a$ at~$p_1$. Assume the
    contrary; then~$t'$ carries~$a$ at~$p_1$. Since~$L$ is closed,~$a$
    must be bound at a position~$p_2$ above~$p_1$ in~$t'$,
    so~$d\nu(t)=d\nu(t')$ is labelled with~$a$ at both~$p_1$
    and~$p_2$. This leaves two cases for position~$p_2$ in~$t$:
    If~$p_2$ is labelled~$a$, then~$a$ must be bound somewhere
    above~$p_2$ in~$t$ because~$t$ is closed; otherwise,~$a$ is bound
    at~$p_2$ in~$t$. Both are in contradiction with~$t$ being clean.
  \end{proof}

\subsubsection*{Proof of \autoref{th:litDec}}
\begin{proof}
As mentioned in the sketch, the proof is analogous to the proof of \autoref{th:dec}.

\begin{enumerate}

\item We show that $D(L_\alpha(A))\subseteq D(L_\alpha(B))$ iff
$L(A)\cap\Terms_S(\Sigma)\subseteq{\downarrow}(L(B_\bot)\cap\Terms_S(\Sigma))$.
Again, $S$ is a finite set of $d_A \cdot n_{ar} + 1$ names.

"$\Rightarrow$": Let $D(L_\alpha(A))\subseteq D(L_\alpha(B))$ and
$t \in L(A)\cap\Terms_S(\Sigma)$,
we show that $t \in {\downarrow}(L(B_\bot)\cap\Terms_S(\Sigma))$.
Since $D(L_\alpha(A))\subseteq D(L_\alpha(B))$, we have by
\autoref{lem:data-inclusion} some $t'$ such that $t \sqsubseteq t'$
and $[t'] \in L_\alpha(B)$. Then by \autoref{th:nameDropAlpha},
$t' \in L(B_\bot)$ and consequently by the definition of $\downarrow$
holds $t \in {\downarrow}(L(B_\bot)\cap\Terms_S(\Sigma))$.

"$\Leftarrow$": Let $L(A)\cap\Terms_S(\Sigma)\subseteq
{\downarrow}(L(B_\bot)\cap\Terms_S(\Sigma))$ and $d\nu(t) \in \lfresh(L_\alpha(A))$,
we show that $d\nu(t) \in \lfresh(L_\alpha(B))$. By the definition of $\lfresh$,
$[t]_\alpha \in L_\alpha(A)$, $t' \in L(A)$ for some $t' \equiv_\alpha t$
and by \autoref{lem:restrictNames} $t'' \in (L(A) \cap \Terms_S(\Sigma))$
for some $t'' \equiv_\alpha t$. By hypothesis, $t'' \in
{\downarrow}(L(B_\bot)\cap\Terms_S(\Sigma))$, then in particular $t''
\in L(B_\bot)$, hence $[t]_\alpha \in L(B)$ and eventually $d\nu(t) \in \lfresh(L_\alpha(B))$
by definition.

\item The construction of the NFTAs~$A_S$ and~$B_S$ is essentially the
  same as in Step~\ref{item:nfta} in the proof of \autoref{th:dec}.
  The only modification is that for every rewrite rule
  $\delta = q(\nu a.f(x_1,\dots, x_n)) \to \nu a.f(q_1(x_1),\dots,
  q_n(x_n))$ in~$B_S$, we add another rewrite rule
  $\delta' = q(a.f(x_1,\dots, x_n)) \to a.f(q_1(x_1),\dots, q_n(x_n))$
  to close the accepted language under $\sqsubseteq$, thus at most
  doubling the number of rewrite rules in~$B_S$.

\item Since the closure under $\sqsubseteq$ increases the size of the
  NFTAs~$A_S$ and~$B_S$ by at most a constant factor~$2$, the size
  estimates in Step~\ref{item:size} in the proof of \autoref{th:dec}
  hold analogously.\qedhere

\end{enumerate}

\end{proof}

\end{document}